\documentclass[journal]{IEEEtran}
\IEEEoverridecommandlockouts                              
\usepackage{graphicx} 
\graphicspath{{figures/}}
\usepackage{amsmath} 
\usepackage{amssymb}  
\usepackage{amsthm}
\usepackage{lettrine}
\usepackage{epstopdf}

\usepackage{xcolor}
\usepackage{gensymb}
\usepackage{soul}
\usepackage{pgfplots}
\pgfplotsset{compat=newest} 
\pgfplotsset{plot coordinates/math parser=false} 

\usepackage[absolute,overlay]{textpos}
\usepackage{mathtools}
\usepackage{algorithm}
\usepackage{algpseudocode}
\usepackage{subfig}
\usepackage{verbatim}
\usepackage{multirow}
\usepackage{float}
\usepackage{enumerate}
\usepackage{soul}
\usepackage[font=small]{caption}
\usepackage{cite}
\mathtoolsset{showonlyrefs} 
\noeqref{www,wwww,eq:KKT_PEVs_stat,eq:inf_inner_a,eq:inf_inner_b,eq:inf_inner_c,eq:F_N,eq,eq:apa_inner_a,eq:conclusion_proof_APA}

\newcommand{\defeq}{\coloneqq}
\newcommand{\eqdef}{\eqqcolon}
\definecolor{QuasiBlue}{rgb}{0.03,0.3,0.72} 

\definecolor{PortlandGreen}{RGB}{99,166,63} 
\definecolor{ReadableGreen}{RGB}{0,0,0}
\newcommand{\basi}{\textcolor{black}}
\definecolor{OrangeRed}{RGB}{255,69,0} 

\definecolor{QuasiBlue}{rgb}{0.03,0.3,0.72}

\newcommand{\GN}{\mc{G}_\N}
\newcommand{\GNS}{(\GN)_{\N=1}^\infty}

\definecolor{Granata}{rgb}{0.64,0,0} 

\newcommand{\di}{d}
\newcommand{\N}{M}

  \newenvironment{sma}
  {\left[\begin{smallmatrix}}
  {\end{smallmatrix}\right]}

\newtheorem{definition}{Definition}
\newtheorem{theorem}{Theorem}
\newtheorem{corollary}{Corollary}
\newtheorem{lemma}{Lemma}

\newtheorem{remark}{Remark}
\newtheorem{proposition}{Proposition}
\newtheorem{assumption}{Assumption}

\newcommand{\smallblacksquare}{\mbox{\rule[0pt]{1.3ex}{1.3ex}}}

\newcommand{\eval}{_{|{z=\sigma(x)}}}

\newcommand{\changednames}[1]{#1}

\usepackage{amsfonts}
\DeclareSymbolFont{bbold}{U}{bbold}{m}{n}
\DeclareSymbolFontAlphabet{\mathbbold}{bbold}

\newcommand{\vect}[1]{\mathbbold{#1}}
\newcommand{\zeros}[1][]{\vect{0}_{#1}}
\newcommand{\ones}[1][]{\vect{1}_{#1}}
\newcommand{\onestight}[1][]{\vect{1}_{\! #1}}
\newcommand{\Dx}{R}

\newcommand{\emm}{_{_M}}
\newcommand\zero{^0}

\newcommand\mar[1]{\textcolor{black}{#1}}
\newcommand\fraa[1]{\textcolor{black}{#1}}
\newcommand\mc[1]{\mathcal{#1}}
\newcommand\mb[1]{\mathbb{#1}}
\newcommand\R{\mathbb{R}}

\renewcommand\i{^i}
\renewcommand\j{^j}
\newcommand\mi{^{-i}}
\newcommand\WE{_\textup{W}}
\newcommand{\VWE}[1]{\bar{#1}_\textup{W}}
\newcommand\NE{_\textup{N}}
\newcommand{\VNE}[1]{\bar{#1}_\textup{N}}

\newcommand{\argmin}[1]{\underset{#1}{\operatorname{argmin}}\,}
\newcommand{\minn}[1]{\underset{#1}{\operatorname{min}}\,}
\newcommand{\maxx}[1]{\underset{#1}{\operatorname{max}}\,}

\newcommand\myqed{{\hspace*{\fill}~\QEDopen\endtrivlist\unskip}}

\newcommand\Wardrop{Wardrop }

\title{\LARGE \bf
Nash and \Wardrop equilibria \\ in aggregative games with coupling constraints
}

\author{\changednames{Dario Paccagnan$^*$, Basilio Gentile$^*$, Francesca Parise$^*$}, Maryam Kamgarpour and John Lygeros
\thanks{$^*$These authors contributed equally to this work.
This work was supported by the European Commission project DYMASOS (FP7-ICT 611281), by the Swiss Competence Centers for Energy Research FEEB\&D, by the ERC Starting Grant CONENE and by the SNSF grant number P2EZP2\_168812.
\changednames{D. Paccagnan, B. Gentile,} M. Kamgarpour and J. Lygeros are with the Automatic Control Laboratory, ETH Z\"{u}rich, Switzerland. {\tt\footnotesize\{dariop,gentileb,mkamgar,lygeros\}@control.ee.ethz.ch.}
F. Parise is with the Laboratory for Information and Decision Systems, MIT, Cambridge, MA, USA. {\tt\footnotesize parisef@mit.edu}
The authors thank Marius Schmitt for fruitful discussion on the traffic model of Section~\ref{sec:traffic} and Ricardo Campos for helping with Lemma~\ref{lem:min_eigenval}.
}}

\begin{document}
\setstcolor{PortlandGreen}
\setul{}{1.5pt}

\maketitle
\thispagestyle{plain}
\pagestyle{plain}

\begin{abstract}
We consider the framework of aggregative games, in which the cost function of each agent depends on his own strategy and on the average population strategy. As first contribution, we investigate the relations between the concepts of Nash and Wardrop equilibria. By exploiting a characterization of the two equilibria as solutions of variational inequalities, we bound their distance with a decreasing function of the population size. As second contribution, we propose two decentralized algorithms that converge to such equilibria and are capable of coping  with constraints coupling the strategies of different agents. Finally, we study the applications of charging of electric vehicles and of route choice on a road network.
\end{abstract}

\vspace*{-0.4cm}

%
\section{Introduction}
\label{sec:intro}
\lettrine{C}{omplex} systems resulting from the interconnection of selfish agents have attracted an increasing interest in the scientific community over the last decade
for their ubiquitous appearance in real-life applications and the inherent mathematical challenges that they present.
Among the vast literature of non-cooperative game theory, \textit{aggregative games}~\cite{jensen2010aggregative} describe systems where each agent is not subject to a one-to-one interaction, but is rather influenced by an aggregate quantity depending on the strategies of the entire population.
The vast spectrum of their applications ranges from traffic~\cite{smith1979existence} or transmission networks \cite{alpcan2002cdma} to electricity \cite{chen2014autonomous} or commodity markets \cite{cournot1838recherches}.
Extending our preliminary work~\cite{paccagnan2016distributed}, we focus on aggregative games where the aggregate quantity is the average population strategy.  Specifically, we address three aspects which are discussed in detail in the next subsections.
\vspace*{-0.2cm}
\subsection*{Nash and Wardrop equilibria}
\label{subsec:intro_Nash_Wardrop}
A fundamental concept in game theory is the notion of \textit{Nash equilibrium}, which is a set of strategies where no agent can lower his cost by unilaterally altering his strategy. Note that in aggregative games an agent can indirectly influence his cost through his contribution to the average strategy. However, when the population becomes large, such contribution becomes negligible. This consideration motivates the introduction of the \textit{Wardrop equilibrium}, which describes a configuration where no agent can lower his cost by altering his strategy, under the assumption that he has no influence on the average.
While the notion of Nash equilibrium has been applied to a large class of problems (see e.g.~\cite{cournot1838recherches,von2007theory} in economics and \cite{alpcan2002cdma} in communication networks), the concept of Wardrop equilibrium is typically formulated in the settings of congestion games or routing problems (see e.g.~\cite{haurie1985relationship} in network congestion games, \cite{wardrop1952road} in road networks, \cite{ma2013decentralized} in electricity markets, \cite{Dafermos87} in economics).
The overarching goal of the first part of this manuscript is to extend the concept of Wardrop equilibrium to generic aggregative games, and to highlight the fundamental connections between Nash and Wardrop equilibria within this setting.
More in details, we leverage on the theory of variational inequality~\cite{facchinei2007finite,Harker91} to
\begin{itemize}
\item[-] present a unifying framework to characterize both Nash and Wardrop equilibria for generic aggregative games, 
\item[-] sharpen the intuition that in large aggregative games Nash and Wardrop strategies are close by bounding their Euclidean distance with a decreasing  function of the population size.
\end{itemize}
We note that the  relation between  Nash and Wardrop equilibria has been extensively studied in the literature, see e.g.,\cite{grammatico:parise:colombino:lygeros:14,altman2004equilibrium,altman2006survey,altman2011routing,haurie1985relationship,Dafermos87} and references therein. We provide a detailed comparison in Section \ref{sec:comparison}, where we show that our contribution significantly differs from the works above.
Our results require the strategy sets of the agents to be uniformly bounded, thus excluding unlimited growth in one or more components of the agents' state space. This assumption is justified by real world applications such as charging of electric vehicles or traffic coordination, as detailed in Sections \ref{sec:traffic}, \ref{sec:PEVs}. Therein the charging requirement of each vehicle or its travel demand are bounded and independent from the rest of the population.

We further note our work proceeds in a similar spirit as in the theory of mean-field games~\cite{lasry2007mean,huang2007large}. Indeed,  both in aggregative and in mean-field games the agents are influenced only by the aggregate population behavior. Consequently, the contribution of a single agent  to  the cost of the other agents  becomes negligible as the population size increases. There are however some important differences between these two classes of games, so that neither is a subset of the other. Specifically, mean-field games are dynamic stochastic games, while our setup is deterministic and static.\footnote{We note that dynamic games over finite horizon can be reformulated in terms of multi-dimensional static games.
This allows us to consider heterogeneous agents with personalized individual and coupling constraints, which cannot be handled in the mean-field game setup. As a consequence, the results typically derived in mean-field games cannot be applied in our setup. 
We further note that these works do not investigate the Euclidean distance between the equilibrium strategies.}

\vspace*{-0.2cm}
\subsection*{Decentralized algorithms and coupling constraints}
\label{subsec:intro_dec_coupl}
The second part of the paper focuses on coordinating  the agents to a Nash or a Wardrop equilibrium for populations of any size \mar{(not necessarily large)}, in the presence of  constraints coupling the agents' strategies. As discussed in the seminal work \cite{rosen1965existence}, when the agents are subject to a coupling constraint, one in general should expect a manifold of equilibria. Here we focus on the specific subclass of  \textit{variational equilibria} \cite{facchinei2007generalized}, which intuitively corresponds to an equal split of the coupling constraint burden among the agents (see Section \ref{sec:norm}).  

Contrary to the classic game theoretical  literature  on coupling constraints (see e.g. \cite{rosen1965existence,Harker91} and references therein), we focus here on deriving  equilibrium coordination algorithms that can be implemented in a \textit{decentralized} fashion. This new requirement is motivated by  reasons of privacy as well as computational intractability of centralized solutions in large scale systems. Specifically, we assume that each agent only knows its own cost function, its individual constraints and its contribution to the coupling constraint. Coordination is achieved by iterative communications with a central coordinator, that can gather and broadcast signals to the population.
Following the recent literature on decentralized coordination for games without coupling constraints, we consider two different scenarios based on whether the agents respond to the common signal by solving a minimization problem (\textit{optimal response}) as in \cite{grammatico:parise:colombino:lygeros:14,parise2015network} or by taking a \textit{gradient step} as in \cite{dario2015aggregative,koshal2012gossip}. Differently from all the aforementioned works, we however  consider constraints coupling the agents' decisions. Specifically, building upon~\cite{facchinei2007generalized}, we contribute as follows:
\begin{itemize}
\item[-] we propose a decentralized two-level algorithm based on optimal response,
which integrates the scheme proposed in~\cite{grammatico:parise:colombino:lygeros:14} with an outer loop that updates a dual variable to achieve a Wardrop equilibrium; 
\item[-] we propose a decentralized one-level asymmetric projection algorithm based on gradient step to achieve either a Nash or a Wardrop equilibrium.
\end{itemize}

While coupling constraints are of fundamental importance in  technical applications, such as electricity markets~\cite{kar2012distributed}, or communication networks~\cite{pan2009games}, we are not aware of previous decentralized coordination schemes that take them into account within the literature of aggregative games. Distributed algorithms for generic games with coupling constraints have been recently suggested in \cite{frazzoli,yin2011nash,yi2017distributed}. In the context of aggregative games these algorithms  however require bilateral communications among all the agents, thus limiting their applicability  in large population games. The  algorithms  for the  case without coupling constraints build on the core assumption that the strategy sets are decoupled and thus cannot be easily adapted to handle coupling constraints. We overcome these difficulties by introducing a dual variable associated with the coupling constraint, which is broadcasted by the central operator, so that each agent reacts to an extended cost function (with an additional price to pay when the coupling constraint is violated) but has decoupled strategy sets. We guarantee that at convergence we reach not only an equilibrium satisfying the coupling constraints of such extended game, as in \cite{Grammatico2016Arxiv}, but indeed a (generalized) equilibrium of the original game.

Outside the game theoretical framework, our algorithms connect with those in~\cite{koshal2011multiuser} for multi-user optimization, where however the agents do not \mbox{influence the cost of the others.}
\subsection*{Applications}
\subsubsection*{Charging of Electric Vehicles}
Electric-vehicles (EV) are foreseen to significantly penetrate the market in the coming years~\cite{nemry2010plug}, therefore coordinating their charging schedules can provide services beneficial to the grid operations~\cite{gan2013optimal}. By assuming that the electricity price depends on the aggregate consumption,~\cite{ma2013decentralized,grammatico:parise:colombino:lygeros:14,dario2015aggregative} formulate the EV charging problem as an aggregative game and propose decentralized schemes based on optimal response or gradient step, in the absence of coupling constraints. The proposed schemes steer the population to Nash \cite{dario2015aggregative} or Wardrop \cite{ma2013decentralized,grammatico:parise:colombino:lygeros:14} equilibria.
We extend the existing literature by introducing constraints coupling the agents' charging profiles. Such constraints model limits on the aggregate peak consumption or on the local consumption of EVs connected to the same transformer. We exploit our theoretical findings to derive results specific to the EV game. Finally, we establish uniqueness of the dual variables associated to the violation of the coupling constraints.
\subsubsection*{Route choice on a road network}

Traffic congestion is a well-recognized issue in densely populated cities, and the corresponding economic costs are significant~\cite{arnott1994economics}. Since every driver seeks his own interest (e.g., minimizing the travel time) and is affected by the others' choices via congestion, a classic approach is to model the traffic problem as a game~\cite{dafermos1980traffic}. Specializing~\cite[Section 1.4.5]{facchinei2007finite}, we focus on a stationary model that aims at capturing the basic interactions among the vehicles flow during rush hours. Building upon our theoretical findings, we derive results specific for the route choice game. Moreover, we perform a realistic numerical analysis based on the data set of the city of Oldenburg in Germany~\cite{OldenburgDataset}. Specifically, we investigate via simulation the effect of road access limitations, expressed as coupling constraints~\cite{sandmo1975pigovian}.
\\[2mm]
\textit{Organization:} 
Sections~\ref{sec:game} and \ref{sec:connection_VI} introduce game and preliminary results. Sections~\ref{sec:Wardrop_Nash} and \ref{sec:algorithms} present our main contributions, namely the bound on the distance between Nash and Wardrop equilibria and the design of decentralized algorithms to achieve them. Sections~\ref{sec:PEVs} and~\ref{sec:traffic} focus on the applications.

\textit{Notation:} 
\fontdimen2\font=3pt
$\ones[n] \in \R^{n}$ and $\zeros[n] \in \R^{n}$ represent the vectors of unit entries and zero entries, respectively; $e_i$ is the $i^\text{th}$ canonical vector.
Given $A\in\mathbb{R}^{n\times n}$, $A\succ0$ ($\succeq0$) $\Leftrightarrow$ $x^\top A x>0~(\ge0),$ $\forall x\neq 0$; $\| A \|$ is the induced 2-norm of $A$.
Given $\N$ vectors each in $\R^{n}$, $[x^1;\ldots;x^\N] \defeq [x^i]_{i=1}^\N\defeq[{x^1}^\top,\ldots ,{x^\N}^\top]^\top \in \R^{\N n}$ and $x^{-i}\coloneqq[x_1;\dots;x_{i-1};x_{i+1};\dots;x_\N]\in \R^{(\N-1)n}$.
Given a matrix $A\in\R^{m\times \N n}$, $A_{(:,i)} \in \R^{m \times n}$ is such that $A=[A_{(:,1)},\ldots,A_{(:,\N)}]$.
Given $g(x):\mathbb{R}^n \rightarrow \mathbb{R}^m$ we define $\nabla_x g(x) \in \mathbb{R}^{n\times m}$ with $[\nabla_x g(x)]_{i,j}\coloneqq \frac{\partial g_j(x)}{\partial x\i}$.
Given $g(x):\mathbb{R} \rightarrow \mathbb{R}$, we denote $g'(x) = \frac{\partial g(x)}{\partial x}$.
Given the sets $\mathcal{X}^1,\dots, \mathcal{X}^{\N} \subseteq \R^n$, we denote $\frac{1}{\N} \sum_{i=1}^\N \mc{X}^i \defeq \{z \in \R^n \vert z = \frac{1}{\N} \sum_{i=1}^{\N} x\i, \text{for some} \: x\i \! \in \! \mc{X}\i \}$.
\section{Problem formulation}
\label{sec:game}

We consider a population of $\N$ agents.
Each agent can choose his strategy $x\i$ in his individual constraint set $\mc{X}\i\subset\mb{R}^n$. We assume that the cost function 
\begin{equation}
J\i(x\i,\sigma(x))
\label{eq:costs_generic}
\end{equation}
of agent $i$ depends on his own strategy $x\i \in \mc{X}\i$ and on the strategies of the other agents via the average population strategy $\sigma(x)\coloneqq\frac{1}{\N}\sum_{j=1}^{\N}x\j \in\frac{1}{\N} \sum_{j=1}^\N \mc{X}\j$, as typical of aggregative games~\cite{jensen2010aggregative}. Besides the individual constraints, each agent has to satisfy a coupling constraint, which involves the decision variables of other agents. Upon defining $x=[x^1;\ldots;x^\N] \in \R^{Mn}$, the coupling constraint can be expressed as
\begin{equation}
x\in \mc{C} \defeq \{x\in\mb{R}^{\N n}\,\vert\,g(x)\le \zeros[m]\} \subset\mb{R}^{\N n},
\label{eq:coupling_constraints_general}
\end{equation}
with $g: \R^{\N n} \to \R^m$. The coupling constraint in~\eqref{eq:coupling_constraints_general} can model for instance the fact that the overall usage level for a certain commodity cannot exceed a fixed capacity. The cost and constraints just introduced give rise to the game
\begin{equation}
 \mc{G} \defeq \left\{ 
 \begin{aligned}
&\! \textup{ agents}:  \; && \{1,\dots,M\}\\
&\! \textup{ cost of agent } i:\quad &&J^i(x\i,\sigma(x)) \\
&\! \textup{ individual constraint} : &&\mc{X}^i\\
&\! \textup{ coupling constraint} :  &&\mc{C},
\end{aligned}\right. 
\label{eq:GNEP}
\end{equation}
which is the focus of the rest of the paper. We denote for convenience $\mc{X} \defeq \mc{X}^1\times\ldots\times\mc{X}^\N$ and define
\begin{equation}\label{eq:Q}
\mc{Q}\i(x\mi)\coloneqq\{x\i\in\mc{X}\i\, \vert \, g(x) \le \zeros[m]\}, \quad \quad \mc{Q} \defeq \mc{X}\cap\mc{C}.
\end{equation}
\subsection{Equilibrium definitions}
We consider two notions of equilibrium for the game $\mc{G}$ in~\eqref{eq:GNEP}. The first  is a known generalization of the concept of Nash equilibrium to games with coupling constraints~\cite{facchinei2007generalized}.
\begin{definition}[Nash Equilibrium]\label{def:NE}
A set of strategies $x\NE = [x^1\NE; \dots; x^\N\NE] \in \R^{\N n}$ is an $\varepsilon$-Nash equilibrium of the game $\mathcal{G}$  if $x\NE\in\mc{Q}$ and for all $ i\in\{1,\dots,\N\}$ and all $ x\i \! \in \! \mc{Q}\i(x\mi\NE)$ 
\begin{align}
 J\i(x\i\NE,\sigma(x\NE)) \! \le \textstyle \! J\i\left( \! x\i,\frac 1\N x\i \! + \! \frac 1\N \sum_{j \neq i} x\j\NE \right) + \varepsilon\,. 
\label{eq:def_GNE}
\end{align}
If~\eqref{eq:def_GNE} holds with $\varepsilon = 0$ then  $x\NE$ is a Nash equilibrium. 
\hfill $\square$
\end{definition}
\noindent Intuitively, a feasible set of strategies $\left\{ x\i\NE \right\}_{i=1}^\N$ is a Nash equilibrium if no agent can improve his cost by unilaterally deviating from his strategy, assuming that the strategies of the other agents are fixed. A Nash equilibrium for a game with coupling constraints is usually referred to as generalized Nash equilibrium~\cite{facchinei2007generalized}; in this paper we omit the word generalized, even though we consider a game with coupling constraints. 

Note that on the right-hand side of~\eqref{eq:def_GNE} the decision variable $x^i$ appears in both arguments of $J^i(\cdot,\cdot)$. However, as the population size grows the contribution of agent $i$ to $\sigma(x)$ decreases. This motivates the definition of Wardrop equilibrium.
\begin{definition}[Wardrop Equilibrium]\label{def:WE}
A set of strategies $x\WE = [x^1\WE; \dots; x^\N\WE] \in \R^{\N n}$ is a Wardrop equilibrium of the game $\mathcal{G}$  if $x\WE\in\mc{Q}$ and for all $i\in\{1,\dots,\N\}$ and all $x\i \! \in \! \mc{Q}\i(x\mi\WE)$
\begin{equation}
 \hspace{0.7cm}J\i(x\i\WE,\sigma(x\WE)) \le J\i ( x\i,\sigma(x\WE) ). \tag*{\qed} 
\end{equation}
\end{definition}
\noindent 
Intuitively, a feasible set of strategies $\left\{ x\i\WE \right\}_{i=1}^{\N}$ is a Wardrop equilibrium if no agent can improve his cost by unilaterally deviating from his strategy, assuming that the average strategy is fixed.
Even though the Wardrop equilibrium is a classical concept, the existing literature on aggregative games \cite{altman2002nash, altman2004equilibrium, marcotte1995convergence, Dafermos87} defines the latter equilibrium in terms of $\sigma(x)$, whereas Definition \ref{def:WE} is expressed in terms of the agents' strategies $x$. The first glimmer of Wardrop equilibrium in terms of strategies appears in \cite{ma2013decentralized,grammatico:parise:colombino:lygeros:14}, where however it is not recognized as an equilibrium concept on its own, but rather only identified as an $\varepsilon$-Nash.
\section{Connection with variational inequalities}
\label{sec:connection_VI} 
This section shows that some equilibria of the game $\mathcal{G}$ in~\eqref{eq:GNEP} can be obtained by solving a variational inequality. This fact is then used to derive the results of Sections~\ref{sec:Wardrop_Nash} and~\ref{sec:algorithms}.
\begin{definition}[Variational inequality~\cite{facchinei2007finite}]
\label{def:vi}
Consider a set $\mathcal{K}\subseteq \R^\di$ and an operator $F:\mc{K} \rightarrow \R^\di$. A point $\bar x\in\mathcal{K}$ is a solution of the variational inequality $\textup{VI}(\mathcal{K},F)$ if 
\begin{equation}
F(\bar x)^\top (x-\bar x)\ge 0, \quad \forall x\in\mathcal{K}. \tag*{\qed}
\end{equation}
\end{definition} 
\noindent Let us define 
\begin{subequations}
\label{eq:F}
\begin{align}
\label{eq:F_N}
F\NE(x)&\defeq[ \nabla_{x\i} J^i(x\i,\sigma(x)) ]_{i=1}^\N\,,\\
F\WE(x)&\defeq [ \nabla_{x\i} J^i(x\i,z)_{|{z=\sigma(x)}} ]_{i=1}^\N \,, \label{eq:F_W}
\end{align}
\end{subequations}
where $F\NE,~F\WE : \mc{X} \rightarrow \mathbb{R}^{\N n}.$
The operator $F\NE$ is obtained by stacking together the gradients of each agent's cost with respect to his decision variable. $F\WE$ is obtained similarly, but considering $\sigma(x)$ as fixed when differentiating. The following proposition provides a sufficient characterization of the equilibria described in Definitions~\ref{def:NE} and \ref{def:WE} as solutions of two variational inequalities, which feature the same set $\mathcal{Q}$, defined in~\eqref{eq:Q}, but different operators, namely $F\NE$ and $F\WE$ in~\eqref{eq:F}. We note that this characterization of equilibria is equivalent to the one in terms of fixed point of the best response mappings, since any variational inequality can be equivalently characterized as a fixed point problem, as explained in \cite[paragraph 12.1.1]{facchinei2007finite}.
\begin{assumption}
\label{A1}
For all $i\in\{1,\dots,\N\}$, the constraint set $\mathcal{X}\i$ is closed and convex.
The set $\mc{Q}$ in \eqref{eq:Q} is non-empty.
The cost functions $J\i(x\i,\sigma(x))$ are convex in $x\i$ for any fixed $\{x^j\in\mc{X}^{j}\}_{j\neq i}$.
The cost functions $J^i(x\i,z)$ are convex in $x\i$ for any $z\in \frac{1}{\N} \sum_{j=1}^\N \mc{X}^j$.
The cost functions $J^i(z_1,z_2)$ are continuously differentiable in $[z_1;z_2]$ for any $z_1 \in \mathcal{X}^i$ and $z_2\in \frac{1}{\N} \sum_{j=1}^\N \mc{X}^j$.
The function $g$ in~\eqref{eq:coupling_constraints_general} is convex.
\myqed
\end{assumption}
\begin{proposition}\label{prop:vi_ref}
Under Assumption~\ref{A1}, the following hold.
\begin{enumerate}
\item Any solution $\VNE{x}$ of VI$(\mathcal{Q},F\NE)$ is a Nash equilibrium of the game $\mc{G}$ in~\eqref{eq:GNEP};
\item Any solution  $\VWE{x}$ of VI$(\mathcal{Q},F\WE)$  is a Wardrop equilibrium of the game $\mc{G}$ in~\eqref{eq:GNEP}. \myqed
\end{enumerate}
\end{proposition}
\begin{proof}
The proof of the first statement can be found in~\cite[Theorem 2.1]{facchinei2007generalized_2}, we prove the second one. We rewrite the operator $F\WE(x)$ as $\tilde F\WE(x,\sigma(x))$, where $\tilde F\WE(x,z)\defeq [ \nabla_{x\i} J^i(x\i,z) ]_{i=1}^\N.$ Fix $ \bar{z}=\sigma(\VWE{x})$. By definition, if $\VWE{x}$ solves VI$(\mathcal{Q},F\WE)$ then $F\WE(\VWE{x} )^\top(x-\VWE{x} )\ge 0$ for all $x\in\mc{Q}$, i.e.
\begin{equation}
\tilde F\WE(\VWE{x}, \bar{z} )^\top(x-\VWE{x} )\ge 0, \ \forall x\in\mc{Q}.
\label{eq:proof_intermediate}
\end{equation}
Consider $i \in \{1,\dots,\N\}$, set $x\mi=\VWE{x}\mi$ in~\eqref{eq:proof_intermediate} and consider an arbitrary $x^{i} \in \mc{Q}^i(\VWE{x}^{-i})$; then all the summands in~\eqref{eq:proof_intermediate} vanish except the $i^{\textup{th}}$ one and~\eqref{eq:proof_intermediate} reads
\begin{equation}
\nabla_{x\i} J\i(\VWE{x}^i,\bar{z})^\top (x\i-\VWE{x}\i)\ge 0, \ \forall \ x\i\in\mc{Q}\i(\VWE{x}\mi).
\label{eq:minimum_principle}
\end{equation}  
Consider the convex function $J\i(\cdot,\bar{z}):\mc{Q}\i(\VWE{x}\mi) \rightarrow \R$. Since $\mc{Q}\i(\VWE{x}\mi)$ is a convex set, by~\eqref{eq:minimum_principle} and~\cite[Proposition 3.1]{tsitsiklis1989parallel} we have that $\VWE{x}^i\in \arg \! \min_{x^i\in\mc{Q}\i(\VWE{x}\mi) } J\i\left(x\i,\bar{z} \right)$.
Substituting $\bar{z}=\sigma(\VWE{x})$, one has $J\i\left(\VWE{x}^i,\sigma(\VWE{x} )\right) \le J\i\left(x\i,\sigma(\VWE{x} )\right)$ for all $x\i\in\mc{Q}\i(\VWE{x}\mi)$. Since this holds for all $i\in \{1,\dots,\N\}$ and since $\VWE{x}\in\mathcal{Q}$, it follows that  $\VWE{x}$ is a Wardrop equilibrium of $\mc{G}$.
\end{proof}
Proposition~\ref{prop:vi_ref} states that a solution of the variational inequality is an equilibrium. The converse in general does not hold due to the presence of the coupling constraints. If on the other hand $\mathcal{C}=\R^{\N n}$, then $\mathcal Q = \mathcal X$ and one can show that $x\NE$ solves the VI$(\mathcal{Q},F\NE)$ if and only if it is a Nash equilibrium of $\mc{G}$ and $x\WE$ solves the VI$(\mathcal{Q},F\WE)$ if and only if it is a Wardrop equilibrium of $\mc{G}$~\cite[Corollary 1]{facchinei2007generalized}.
The equilibria that can be obtained as solution of the corresponding variational inequality are called \textit{variational equilibria}~\cite[Definition 3]{facchinei2007generalized} and are here denoted with $\VNE{x},\VWE{x}$ instead of $x\NE,x\WE$ (indicating any equilibria satisfying Definitions~\ref{def:NE} and~\ref{def:WE}). We next provide sufficient conditions for existence and uniqueness of variational~equilibria.
\begin{definition}[Strong monotonicity~\cite{facchinei2007finite}]
\label{def:SMON}
An operator $F: \mc{K} \subseteq \R^\di \rightarrow \R^\di$ is strongly monotone on the set $\hat{\mc{K}} \subseteq \mc{K}$ with monotonicity constant $\alpha>0$ if \footnote{When we do not specify the set $
\hat{\mc{K}}$ this is understood to be $\mc{K}$, i.e. the domain of the operator. Note that in our setup strong monotonicity is equivalent to strict diagonal convexity in~\cite{rosen1965existence}.}
\begin{equation}
(F(x)-F(y))^\top(x-y) \ge \alpha \|x-y\|^2,\quad \forall x,y \in \hat{\mc{K}}.
\label{eq:def_strongly monotone}
\end{equation}
The operator is monotone on $\hat{\mc{K}}$ if~\eqref{eq:def_strongly monotone} holds for $\alpha=0$.
\hfill{$\square$}
\end{definition}
\begin{lemma}\textup{\cite[Corollary 2.2.5, Theorem 2.3.3]{facchinei2007finite}}
\label{lem:exun}
Let Assumption~\ref{A1} hold. Then
\begin{enumerate}
\item If $\mc{X}^i$ is bounded for all $i \in \{1,\dots,\N\}$, then both $\textup{VI}(\mathcal{Q},F\NE)$ and $\textup{VI}(\mathcal{Q},F\WE)$ admit a solution\footnote{The convexity of the cost functions required by Assumption~\ref{A1} is not needed for the first statement of Lemma~\ref{lem:exun}, continuity is enough.}.
\item If $F\NE$ is strongly monotone on $\mc{Q}$, then $\textup{VI}(\mathcal{Q},F\NE)$ has a unique solution. If $F\WE$ is strongly monotone on $\mc{Q}$ then $\textup{VI}(\mathcal{Q},F\WE)$ has a unique solution. \hfill{$\square$}
\end{enumerate}
\end{lemma}
\subsection{Variational  and normalized equilibria}
\label{sec:norm}
The concept of games with coupling constraints  has first been introduced in the seminal work \cite{rosen1965existence}. Therein the key concept of \textit{normalized equilibria} has been introduced to describe the fact that when the agents are subject to a coupling constraint, even under
strong monotonicity conditions, one should expect a manifold of equilibria. 
Formally, the vector $x\NE$  is a normalized Nash equilibrium if there exists a vector of weights $r\in\R^M_{\ge0}$, with $\sum_{i=1}^M r_i =1$, such that $x\NE$ solves the VI$(\mathcal{Q},F\NE^r )$ where $F\NE^r(x)\defeq[ r_i\nabla_{x\i} J^i(x\i,\sigma(x)) ]_{i=1}^\N$.
It is proven in \cite{rosen1965existence} that the choice of $r$ corresponds to a split of the burden of satisfying the constraints among the agents.
In the context of aggregative games, however, each agent contributes equally to the average. Therefore it is typically assumed that the burden of the constraint should also be split equally among the agents by selecting $r=\frac1 M\ones[M]$, see e.g.,~\cite{facchinei2007generalized,pan2009games,facchinei2007generalized_2}. It is immediate to see that the subclass of normalized equilibria for which this property holds is the class of \textit{variational equilibria} (introduced in the previous section) and is the one on which we focus from here on. Nonetheless we note that our results could be easily extended to normalized equilibria by using operator $F\NE^r$ instead of $F\NE$. Similar arguments hold for the Wardrop~equilibrium.
\subsection{Sufficient conditions for monotonicity}
\noindent
To verify whether an operator is strongly monotone or monotone one can exploit the following equivalent characterizations.
\begin{lemma}
\label{lemma:pd}
\textup{\cite[Proposition 2.3.2]{facchinei2007finite}}
A continuously differentiable operator $F: \mc{K} \subseteq \R^\di \to \R^\di$ is strongly monotone with monotonicity constant $\alpha$ (resp. monotone) if and only if $\nabla_x F(x)\succeq \alpha I$ (resp. $\nabla_x F(x)\succeq 0$) for all $x \in \mc{K}$. Moreover, if $\mc{K}$ is compact then there exists $\alpha>0$ such that $\nabla_x F(x)\succeq \alpha I$ for all $x \in \mc{K}$ if and only if $\nabla_x F(x)\succ 0$ for all $x \in \mc{K}$.
\myqed
\end{lemma}
The previous lemma can be used to derive sufficient conditions for strong monotonicity of the operators $F\NE$ and $F\WE$.
To this end, we specialize in the subsequent Lemma~\ref{lem:FNstrongly monotone} the cost function~\eqref{eq:costs_generic} of agent $i$ to
\begin{equation} J\i(x\i,\sigma(x)) \defeq v\i(x\i) +p(\sigma(x))^\top x\i.
\label{eq:costs_specific}
\end{equation} 
The cost in~\eqref{eq:costs_specific} can for example describe applications where $x^i$ denotes the usage level of a certain commodity, whose negative utility is modeled by $v^i: \mc{X}\i \to \R$ and whose per-unit cost $p : \frac{1}{\N} \sum_{i=1}^\N \mc{X}^i \to \R^n$ depends on the average usage level of the entire population~\cite{chen2014autonomous,ma2013decentralized}.
The operators in~\eqref{eq:F}  become
\begin{subequations}
\label{eq:F_decomp_specific}
\begin{align}
F\WE(x) &= [\nabla_{x\i} v\i(x\i)]_{i=1}^\N + [p(\sigma(x))]_{i=1}^\N, \label{eq:F_W_decomp} \\
F\NE(x) &= \textstyle F\WE(x) + \frac1 \N [  \nabla_z p(z)_{|{z=\sigma(x)}} {x\i}  ]_{i=1}^\N. \label{eq:F_N_decomp}
\end{align}
\end{subequations}
\begin{lemma}\label{lem:FNstrongly monotone}
\leavevmode
\begin{enumerate}
\item Suppose that for each agent $i\in\{1,\dots,\N\}$ the function $v^i$ in~\eqref{eq:costs_specific} is convex and that $p$ is monotone; then $F\WE$ is monotone. Under the further assumption that $p$ is affine and strongly monotone, $F\NE$ is strongly monotone.
\item Suppose that for each agent $i\in\{1,\dots,\N\}$ the function $v^i$ in~\eqref{eq:costs_specific} is strongly convex and that $p$ is monotone. Then $F\WE$ is strongly monotone. \qed
\end{enumerate}
\end{lemma}
The proof is reported in the Appendix.
When the price function $p(\sigma)$ has diagonal structure i.e. can be decomposed as $p(\sigma(x))=[p_t(\sigma_t(x_t))]_{t=1}^n$, with $\sigma_t(x_t) = \frac{1}{N}\sum_{i=1}^\N x_t^i$ and $x_t\defeq[x^1_t,\dots,x^n_t]$, it is possible to give additional sufficient conditions that guarantee strong monotonicity of the operators. With this respect, the result in \cite[Theorem 1]{GentileThesis} generalizes the reasoning presented in Corollary \ref{cor:pev} and \ref{cor:traffic1} in Sections \ref{sec:traffic}, \ref{sec:PEVs}.
It is clear from Lemma~\ref{lem:FNstrongly monotone} that often only one of $F\NE$ and $F\WE$ possesses monotonicity properties, which are required to guarantee that an equilibrium can be achieved using the algorithms proposed in Section~\ref{sec:algorithms}. Hence it is important to derive results on the distance between the two equilibria, which is the goal of the next Section~\ref{sec:Wardrop_Nash}.
\section{Distance between Nash and Wardrop equilibria in large populations}
\label{sec:Wardrop_Nash}
In this section we study the relations between Nash and Wardrop equilibria in aggregative games with large populations.
Specifically, we consider a sequence of games $\GNS$.
For fixed $\N$, the game $\GN$ is played among $\N$ agents and is defined as in~\eqref{eq:GNEP} with an arbitrary coupling constraint $\mc{C}$ and, for every agent $i$, \basi{arbitrary $J\i(x\i,\sigma(x))$} and $\mc X\i$. For the sake of readability, we avoid the explicit dependence on $\N$ in denoting these quantities and in denoting $x\NE$, $x\WE$, $F\NE$, $F\WE$.
\begin{assumption}
\label{A2}
There exists a convex, compact set $\mathcal{X}\zero\subset\R^n$ such that $\cup_{i=1}^\N \mathcal{X}^i\subseteq{\mathcal{X}\zero}$ for each $\GN$ in the sequence $\GNS$. For each $\N$ and $i \in \{1,\dots,\N\}$, the function $J\i(z_1,z_2)$ is Lipschitz with respect to $z_2$ in $\mc{X}^0$ with Lipschitz constant $L_2$ independent from $\N$, $i$ and $z_1 \in \mc{X}^i$.
\qed
\end{assumption}
We note that Assumption~\ref{A2} implies that $\sigma(x) \in \mc{X}^0$ for any $\N$ and any $x \in \mc{X}^1 \times \dots \times \mc{X}^\N$. Moreover, under Assumption~\ref{A2} we define $\! \Dx \! \defeq \! \max_{y \in{\mathcal{X}\zero} } \{\|y \|\}$. Furthermore, if the cost function~\eqref{eq:costs_generic} takes the specific form~\eqref{eq:costs_specific}, then $p$ being Lipschitz in $\mc{X}^0$ with constant $L_p$ implies $J\i(z_1,z_2)$ being Lipschitz with respect to $z_2$ in $\mc{X}^0$ with constant $L_2 = \Dx L_p$, as
\begin{equation}
\begin{aligned}
&\| J\i(z_1,z_2) - J\i(z_1,z_2') \| = \| (p(z_2) - p(z_2'))^\top z_1 \| \\
&\le \| p(z_2) - p(z_2') \| \| z_1 \| \le \Dx L_p \| z_2 - z_2' \|.
\end{aligned}
\label{eq:Lipschitz_implies_Lipschitz}
\end{equation}

The next proposition shows that every Wardrop equilibrium is an $\varepsilon$-Nash equilibrium, with $\varepsilon$ tending to zero as $\N$ grows.
\begin{proposition}~\label{prop:conv_cost}
Let the sequence of games $\GNS$ satisfy Assumption~\ref{A2}. For each $\GN$, every Wardrop equilibrium is an $\varepsilon$-Nash equilibrium, with $\varepsilon=\frac{2\Dx L_2}{\N}$.
\qed
\end{proposition}
\begin{proof}
Consider any Wardrop equilibrium $x\WE$ of $\GN$ (not necessarily a variational one). By Definition~\ref{def:WE}, $x\WE\in\mc{Q}$ and for each agent $i$ \[J\i(x\WE^i, \sigma(x\WE)) \le J\i(x^i, \sigma(x\WE)), \quad \textup{for all} \quad x^i\in \mc{Q}\i(x\WE\mi).\] It follows that for each agent $i$ and for all  $x^i\in \mc{Q}\i(x\WE\mi)$
\begin{align}
&\textstyle J\i(x\WE^i, \sigma(x\WE)) - J\i(x^i, \frac{1}{\N}(x\i + \sum_{j\neq i }x\WE^j)) \\
&= \textstyle\underbrace{J\i(x\WE^i, \sigma(x\WE)) - J\i(x^i, \sigma(x\WE))}_{\le 0}+ \\[-0.05cm]
&\textstyle J\i(x^i, \sigma(x\WE)) - J\i(x^i, \frac{1}{\N}(x\i + \sum_{j\neq i }x\WE^j)) \\
& {\color{ReadableGreen}{}\le \textstyle L_2 \| \sigma(x\WE) - (\frac{1}{\N}(x\i + \sum_{j\neq i }x\WE^j)) \| } \\
&{\color{ReadableGreen}{}=\textstyle\frac{L_2}{\N} \|(x\i\WE + \sum_{j\neq i }x\WE^j) - (x\i + \sum_{j\neq i }x\WE^j) \| \color{red} } \\
&{\color{ReadableGreen}{}= \textstyle \frac{L_2}{\N} \| x\i\WE - x\i  \| \le \frac{2\Dx L_2}{\N}.}
\end{align}
Hence $x\WE$ is an $\varepsilon$-Nash equilibrium of $\GN$.
\end{proof}
Proposition~\ref{prop:conv_cost} is a strong result but it provides no information on the distance between the set of strategies constituting a Nash and the set of strategies constituting a Wardrop equilibrium. In the following we study this distance for variational equilibria.

\begin{theorem}
\label{thm:conv_strategies}
Let the sequence of games $\GNS$ satisfy Assumption~\ref{A2}, and each $\GN$ satisfy Assumption~\ref{A1}. Then:
\begin{enumerate}
\item
If the operator $F\NE$ relative to $\GN$ is strongly monotone on $\mc{Q}$ with monotonicity constant $\alpha\emm>0$, then there exists a unique variational Nash equilibrium $\VNE{x}$ of $\GN$. Moreover, for any variational Wardrop equilibrium $\VWE{x}$
\begin{align}
\| \VNE{x} - \VWE{x} \| \le \frac{L_2}{\alpha\emm{\sqrt \N}}.
\label{eq:convergence_strategies_2}
\end{align}
As a consequence, if $\alpha\emm{\sqrt \N} \to \infty$ as $\N \to \infty$, then $\| \VNE{x} - \VWE{x} \| \to 0$ as $\N \to \infty$.
\item If the operator $F\WE$ relative to $\GN$ is strongly monotone on $\mc{Q}$ with monotonicity constant $\alpha\emm>0$, then there exists a unique variational Wardrop equilibrium $\VWE{x}$ of $\GN$. Moreover, for any variational Nash equilibrium $\VNE{x}$
\begin{align}
\| \VNE{x} - \VWE{x} \| \le \frac{L_2}{\alpha\emm{\sqrt \N}}.
\label{eq:convergence_strategies}
\end{align}
As a consequence, if $\alpha\emm{\sqrt \N} \to \infty$ as $\N \to \infty$,
then $\| \VNE{x} - \VWE{x} \| \to 0$ as $\N \to \infty$.
\item If in each game $\GN$ the cost function $J^i(x^i,\sigma(x))$ takes the form~\eqref{eq:costs_specific}, with $v^i = 0$ and $p$ being strongly monotone on $\mc{X}^0$ with monotonicity constant $\alpha$, then there exists a unique $\bar \sigma$ such that $\sigma(\VWE{x})=\bar \sigma$ for any variational Wardrop equilibrium $\VWE{x}$ of $\GN$. Moreover, for any variational Nash equilibrium $\VNE{x}$ of $\GN$ and for any variational Wardrop equilibrium\footnote{If $p$ is Lipschitz with constant $L_p$, then in~\eqref{eq:convergence_sigma} $L_2$ can be replaced by $R L_p$, as by~\eqref{eq:Lipschitz_implies_Lipschitz}. This is used in the application Sections~\ref{sec:PEVs},~\ref{sec:traffic}.} $\VWE{x}$ of $\GN$
\begin{equation}
\| \sigma(\VNE{x}) - \sigma(\VWE{x}) \| \le  \sqrt{\frac{2\Dx L_2}{\alpha \N}}.
\label{eq:convergence_sigma}
\end{equation}
Hence, $\| \sigma(\bar x\NE) - \sigma(\bar x\WE) \| \to 0$ as $\N \to \infty$. \qed
\end{enumerate}
\end{theorem}
\begin{proof}
1)
We first bound the distance between the operators $F\NE$ and $F\WE$ in terms of $\N$. By~\eqref{eq:F} it holds
\[\begin{split}
&\textstyle \|F\NE(x)-F\WE(x)\|^2 \\
&=\textstyle \| [ \nabla_{x\i} J^i(x\i,\sigma(x)) ]_{i=1}^\N - [\nabla_{x\i} J^i(x\i,z)_{|{z=\sigma(x)}} ]_{i=1}^\N \|^2 \\
&=\textstyle \sum_{i=1}^\N \| \frac{1}{\N} \nabla_{z} J\i(x\i,z)_{|{z=\sigma(x)}} \|^2\fraa{ \le \textstyle \frac{1}{\N^2} \sum_{i=1}^\N L_2^2 = \frac{L_2^2}{\N},}
\end{split}\]
where the inequality follows from the fact that $J^i(z_1,z_2)$ is Lipschitz in $z_2$ on ${\mc{X}\zero}$ with constant $L_2$ by Assumption~\ref{A2} and hence the term $\| \nabla_{z} J\i(x\i,z)_{|{z=\sigma(x)}} \|$ is bounded by $L_2$ by definition of derivative.
It follows that
\begin{equation}
\|F\NE(x)-F\WE(x)\|\le\frac{L_2}{\sqrt{M}}.
\label{eq:distop}
\end{equation}
for all $x \in \mc{X}^0$.
We exploit~\eqref{eq:distop} to bound the distance between Nash and Wardrop strategies. Since $F\NE$ is strongly monotone on $\mc{Q}$ by assumption, $\textup{VI}(\mathcal{Q},F\NE)$ has a unique solution $\VNE{x}$ by Lemma~\ref{lem:exun}. Moreover, by~\cite[Theorem 1.14]{nagurney2013network} for all solutions $\VWE{x} $ of $\textup{VI}(\mc{Q},F\WE)$ it holds
\begin{equation}
\textstyle \|\VNE{x}-\VWE{x}\| \le\frac{1}{\alpha\emm}\|F\NE(\VWE{x})-F\WE(\VWE{x})\|.
\end{equation}
Combining this with equation \eqref{eq:distop} yields the result.
\\
2) As in the above, with Nash in place of Wardrop and viceversa.
\\
3) Any solution $\VWE{x}$ to the $\textup{VI}(Q,F\WE)$ satisfies
\begin{equation}
\begin{aligned}
&  F\WE(\VWE{x})^\top(x-\VWE{x}) \ge 0, \; \forall x \in Q \Leftrightarrow \\
&\textstyle\sum_{i=1}^\N p(\sigma(\VWE{x}))^\top(x\i-\VWE{x}^i) \ge 0, \; \forall x \in Q \Leftrightarrow \\
&p(\sigma(\VWE{x}))^\top(\sigma(x) - \sigma(\VWE{x})) \ge 0, \; \forall x \in Q.
\label{eq:VI_WE_sigma}
\end{aligned}
\end{equation}
Any solution $\VNE{x}$ to the $\textup{VI}(Q,F\NE)$ satisfies
\begin{equation}
\begin{aligned}
& F\NE(\VNE{x})^\top(x-\VNE{x}) \ge 0, \; \forall x \in Q \Leftrightarrow \\
& p(\sigma(\VNE{x}))^\top (\sigma(x) - \sigma(\VNE{x})) + \\
&\frac{1}{\N^2} \sum_{i=1}^\N (\nabla_z p(z)_{|z=\sigma(\VNE{x})} \VNE{x}\i)^\top (x\i - \bar x\i\NE) \ge 0, \; \forall x \in Q .
\label{eq:VI_VNE_sigma}
\end{aligned}
\end{equation}
Exploiting the strong monotonicity of $p$ on $\mc{X}^0$, one has \\
\[\begin{aligned}
& \; \alpha \| \sigma(\VWE{x}) - \sigma(\VNE{x}) \|^2  \\
&\le ( p(\sigma(\VWE{x}))-p(\sigma(\VNE{x})) )^\top (\sigma(\VWE{x})-\sigma(\VNE{x}) ) \\
&={ p(\sigma(\VWE{x}))^\top \! (\sigma(\VWE{x}) \! - \! \sigma(\VNE{x}) )}
\! -  p(\sigma(\VNE{x}))^\top \! (\sigma(\VWE{x}) \! - \! \sigma(\VNE{x})) \\ &\underset{\text{by}~\eqref{eq:VI_WE_sigma}}{\le}
 -  p(\sigma(\VNE{x}))^\top (\sigma(\VWE{x})-\sigma(\VNE{x})) \\
 &\underset{\text{by~\eqref{eq:VI_VNE_sigma}}}{\le} \textstyle \frac{1}{\N^2} \sum_{i=1}^\N (\VNE{x}\i)^\top (\nabla_z p(z)_{|z=\sigma(\VNE{x})})^\top(\VWE{x}^i - \bar x\i\NE) \\
& {\color{ReadableGreen}{} \le \textstyle \frac{1}{\N^2} \sum_{i=1}^\N \|\VNE{x}\i\| \| \nabla_z J\i(\VWE{x}\i,z)_{|z=\sigma(\VNE{x})} \|} \\
& {\color{ReadableGreen}{}+ \textstyle \frac{1}{\N^2} \sum_{i=1}^\N \| \VNE{x}\i\| \| \nabla_z J\i(\VNE{x}\i,z)_{|z=\sigma(\VNE{x})} \|} \\
&{\color{ReadableGreen}{}\le \textstyle  \frac{2 L_2}{\N^2} \sum_{i=1}^\N \|\VNE{x}\i\| \le \textstyle \frac{2L_2}{\N^2} \sum_{i=1}^\N \Dx \le \frac 1\N 2 \Dx L_2.}
\end{aligned}\]
\noindent We  conclude that $\textstyle \| \sigma(\VWE{x}) - \sigma(\VNE{x}) \| \le \sqrt{\frac{2\Dx L_2}{\alpha \N}}.$
\end{proof}
We point out that the bounds~\eqref{eq:convergence_strategies_2} and~\eqref{eq:convergence_strategies} can be used to derive a bound on the average strategies similar to~\eqref{eq:convergence_sigma}.
\subsection{Comparison with the literature}
\label{sec:comparison}
Proposition \ref{prop:conv_cost} states that, under fairly general assumptions, any Wardrop equilibrium is an $\varepsilon$-Nash equilibrium. Such result follows directly from the fact that each agent contributes  only via the average and that the cost functions are  Lipschitz. Consequently, the contribution of each agent scales linearly with the inverse of the population size. This same idea is used to prove  similar results in many previous contributions. For example, the case of  potential games is investigated in  \cite{altman2004equilibrium,altman2006survey},  routing games are considered in \cite{altman2011routing},  flow control and routing in communication networks are discussed in \cite{altman2002nash}, while a similar argument is used in  \cite{grammatico:parise:colombino:lygeros:14} for the case without coupling constraints.   Proposition \ref{prop:conv_cost} is a trivial extension of those works to  generic aggregative games with coupling constraints.

Our main result  is to prove that, by introducing further assumptions, one can actually go beyond Proposition \ref{prop:conv_cost} and derive  bounds on the Euclidean distance between Nash and Wardrop equilibria. In Theorem  \ref{thm:conv_strategies} we consider two types of additional assumptions: the first is strong monotonicity of either the Nash or Wardrop operator (statements 1 and 2), the second is a structural assumption on the cost functions (statement 3). The only previous results bounding the Euclidean distance between the two equilibria that we are aware of are obtained in \cite{haurie1985relationship}. Therein a similar bound to our result of Theorem~\ref{thm:conv_strategies}-3) is derived specific to routing/congestion games. However, that work assumes that the population increases by means of identical replicas of the agents. We here prove that a similar argument as in \cite{haurie1985relationship} can be used to address the case of generic new agents instead of identical copies. Moreover, the results in Theorem~\ref{thm:conv_strategies}-1) and Theorem~\ref{thm:conv_strategies}-2) address a more general class of aggregative games (i.e. not necessarily congestion games) by employing a new  type of argument, based on a sensitivity analysis result for variational inequalities with perturbed  strongly  monotone operators \cite[Theorem 1.14]{nagurney2013network}. We note that the works~\cite{Dafermos87,altman2004equilibrium,altman2006survey} guarantee convergence of Nash to Wardrop in terms of Euclidean distance, but do not provide a bound on the convergence rate.

Finally, our results are derived for variational equilibria. We remark that if there are no coupling constraints, as in the previous  works, then any equilibrium is a  variational equilibrium. Hence our results subsume the results above. We remark that including coupling constraints does not increase the complexity of the mathematical treatment of Section~\ref{sec:Wardrop_Nash}; on the contrary the design of the algorithms in Section~\ref{sec:algorithms} is specifically tailored to account for coupling constraints.

\section{Decentralized algorithms}
\label{sec:algorithms}
In this section we turn our attention to the design of algorithms that achieve a Nash or a Wardrop equilibrium. Hence we do not consider a sequence of games as in the previous section, but rather focus on the game~\eqref{eq:GNEP} with fixed population. We begin with the following assumption on the constraint sets.
\begin{assumption}
The coupling constraint in~\eqref{eq:coupling_constraints_general} is of the form
\begin{equation}
x\in \mc{C} \defeq \{x\in\mb{R}^{\N n}\,\vert\,Ax\le b\} \subset\mb{R}^{\N n},
\label{eq:coupling_constraints_affine}
\end{equation}
with $A \defeq [A_{(:,1)},\ldots,A_{(:,\N)}] \in\R^{m\times{\N n}}$, $A_{(:,i)} \in\R^{m\times{n}}$ for all $i\in \{1,\dots,\N\}$, $b\in\R^m$. Moreover, for all $i \in \{1,\dots,\N\}$, the set $\mc{X}\i$ can be expressed as $\mc{X}\i = \{x\i \in \R^n \vert g\i(x\i) \le 0 \}$, where $g\i: \R^n \to \R^{p_i}$ is continuously differentiable. The set $Q$, which can thus be expressed as $\mc{Q} = \{x \in \R^{Mn} \vert g\i(x\i) \le 0, \: \forall i,\: Ax \le b \}$, satisfies Slater's constraint qualification as by~\cite[(5.27)]{boyd2004convex}. \myqed
\label{A3}
\end{assumption}
We note that linearity of the coupling constraints arises in a range of applications, as explained in~\cite[page 188]{facchinei2007generalized}. We also assume that agent $i$ does not wish to disclose information about his cost function $J\i$ and individual constraint set $\mathcal{X}\i$ and that he knows his influence on the coupling constraint, that is, the sub-matrix $A_{(:,i)}$ in~\eqref{eq:coupling_constraints_affine}. Moreover, we assume the presence of a central operator that is able to measure the population average $\sigma(x)$, to evaluate the quantity $Ax-b$ in~\eqref{eq:coupling_constraints_affine} and to broadcast aggregate information to the agents. Based on this information structure, in the following we focus on the design of decentralized algorithms to obtain a solution of either VI$(\mathcal Q, F\NE)$ or VI$(\mathcal Q, F\WE)$.
As the techniques are the same for Nash and Wardrop equilibrium, we consider the general problem VI$(\mathcal Q, F)$, where $F$ can be replaced with $F\NE$ or $F\WE$.

We observe that, if $F$ is integrable and monotone on $\mc{Q}$, that is, if there exists a convex function $E(x): \R^{\N n}\rightarrow \R$ such that $F(x)=\nabla_x E(x)$  for all $x\in\mc{Q}$, then VI$(\mathcal Q, F)$ is equivalent to the convex optimization problem~\cite[Section 1.3.1]{facchinei2007finite}
\begin{equation}
\argmin{x\in\mathcal{Q}}\ E(x).
\label{eq:opt}
\end{equation}
Therefore a solution of VI$(\mathcal Q, F)$ and thus a variational equilibrium can be found by applying any of the decentralized optimization algorithms available in the literature \cite{tsitsiklis1989parallel} to problem~\eqref{eq:opt}; the decentralized structure arises because each agent can evaluate $\nabla_{x^i} E(x)$ by knowing only his strategy $x\i$ and $\sigma(x)$. Equivalently, the integrability assumption guarantees that $\mc{G}$ is a \textit{potential game} with potential function $E(x)$ \cite{shaply1994potentialgames}, hence decentralized convergence tools available for potential games can also be employed~\cite{dubey2006strategic,marden2009cooperative}. An operator $F$ is integrable in $\mc{Q}$ if and only if \mar{$\nabla_x F(x)=\nabla_x F(x)^\top$} for all $x \in \mc{Q}$~\cite[Theorem 1.3.1]{facchinei2007finite}. We anticipate that in both applications of Sections~\ref{sec:PEVs} and~\ref{sec:traffic} the Wardrop operator $F\WE$ in~\eqref{eq:F_W_decomp} is integrable but the Nash operator $F\NE$ in~\eqref{eq:F_N_decomp} is not.

In the following we intend to find a solution of VI$(\mathcal Q, F)$ when $F$ is not necessarily integrable, so that these standard  methods cannot be applied. To propose decentralized schemes in presence of coupling constraints, we introduce two reformulations of VI$(\mathcal Q, F)$ in an extended space $[x;\lambda]$ where $\lambda$ are the dual variables relative to the coupling constraint $\mc C$. These two reformulations will then be used to propose two alternative algorithms. Specifically, we define for any $\lambda\in\R^m_{\ge0}$ the game
\begin{equation}
\mc{G}(\lambda) \! \defeq \! \left\{ 
\begin{aligned}
&\textup{agents}:& &\{1,\dots,M\}\\
&\textup{cost of agent } i:& &J^i(x^i,\sigma(x)) \!+ \lambda^\top A(:,i)x^i \\
&\textup{individual constr}:& &\mc{X}^i\\
&\textup{coupling constr}:& &\R^{\N n}.
\end{aligned}\right. 
\label{eq:GNEP_ext}
\end{equation}
Moreover, we introduce the extended VI$(\mc{Y},T)$ with
\begin{equation}
\begin{split}
\mathcal{Y} \defeq \mc{X}\times \R^m_{\ge0}\,,\quad 
T(x,\lambda) \defeq
\begin{bmatrix}
F(x)+ A^\top\lambda\\
 -(Ax-b)
\end{bmatrix}\,.
\end{split}
\label{eq:T}
\end{equation}
The following proposition draws a connection between VI$(\mathcal Q, F)$, the game $\mc{G}(\lambda)$ and VI$(\mc{Y},T)$.
\begin{proposition}\textup{\cite[Section 4.3.2]{scutari2012monotone}}
\label{prop:ext_vi}
Let Assumptions~\ref{A1} and~\ref{A3} hold.
The following statements are equivalent.
\begin{enumerate}
\item The vector $\bar x$ is a solution of VI$(\mc{Q},F)$.
\item There exists $\bar \lambda \in \R^m_{\ge 0}$ such that $\bar x$ is a variational equilibrium of $\mc{G}(\bar \lambda)$ and $0\le \bar \lambda \perp b-A\bar x\ge 0$.
\item There exists $\bar \lambda \in \R^m_{\ge 0}$ such that the vector $[\bar x; \bar \lambda]$ is a solution of VI$(\mathcal{Y},T)$. \hfill{$\square$}
\end{enumerate}
\end{proposition}
The proof is an easy adaptation of \cite[Section 4.3.2]{scutari2012monotone} and is postponed to the Appendix. In subsection~\ref{subsec:fully} we exploit the equivalence between 1) and 2) to propose a two-level algorithm based on optimal response that converges to a Wardrop equilibrium. In subsection~\ref{subsec:boundedly} we leverage on the equivalence between 1) and 3) to propose a one-level algorithm based on gradient step that converges to a Nash equilibrium. The same one-level algorithm can be used to obtain a Wardrop equilibrium, by using $F\WE$ instead of $F\NE$.

\subsection{Two-level algorithm based on optimal response for Wardrop equilibrium}
\label{subsec:fully}
Based on the equivalence between 1) and 2) in Proposition~\ref{prop:ext_vi}, we here introduce Algorithm \ref{alg:two} to achieve a Wardrop equilibrium. The algorithm features an outer loop, in which the central operator broadcasts to the population the dual variables $\lambda_{(k)}$ based on the current constraint violation, and an inner loop, in which the agents update their strategies to the Wardrop equilibrium of the game $\mc{G}(\lambda_{(k)})$.
Since $\mc{G}(\lambda_{(k)})$ is a game \textit{without} coupling constraints, the Wardrop equilibrium can be found via the iterative algorithm proposed in~\cite[Algorithm 1]{grammatico:parise:colombino:lygeros:14}. For each agent $i\in\{1,\dots,\N\}$ we define the optimal response to a signal $z\in \frac{1}{\N} \sum_{i=1}^\N \mc{X}\i$ and dual variables $\lambda\in\R^m_{\ge0}$ 
\begin{equation}
x^i_{\textup{or}}(z,\lambda) \defeq \argmin{x^i \in\mc{X}^i} \; J^i(x^i, z )+\lambda^\top A(:,i)x^i.
\label{eq:or}
\end{equation}
\begin{algorithm}
\caption{for Wardrop equilibrium}
\label{alg:two}

{\small
\textbf{Initialization}: Set $k =0$, $\tau > 0$, $x^i_{(0)} \in \R^n$, $\lambda_{(0)}\in\R^m_{\ge0}$.
\vspace{0.1cm}
\textbf{Iterate until convergence}:
\begin{enumerate}
\item \textit{Strategies are updated to  a Wardrop equilibrium of $\mc{G}_{\lambda_{(k)}}$}\\
\hfill{ \begin{minipage}{0.42\textwidth}
\rule{\textwidth}{.4pt}
 \textbf{Initialization}: Set $h =0$, $\tilde x^i_{(0)}=x^i_{(k)}$, $z_{(0)}\in\R^n$.\\[0.1cm]
\textbf{Iterate until convergence}:
\begin{subequations}
\label{eq:inf_inner}
\begin{align}
\label{eq:inf_inner_a}
\tilde  x^i_{(h+1)}&\leftarrow  x^i_{\textup{or}}(z_{(h)},\lambda_{(k)}), \forall i\\
\tilde \sigma_{(h+1)} & \textstyle \leftarrow \frac{1}{\N}\sum_{j=1}^\N \tilde x^j_{(h+1)} \label{eq:inf_inner_b} \\
z_{(h+1)} & \textstyle \leftarrow (1-\frac 1 h) z_{(h)}+\frac 1 h \tilde \sigma_{(h+1)} \label{eq:inf_inner_c} \\
h & \textstyle \leftarrow h+1
\end{align}
\end{subequations}
\textbf{Upon convergence}:
\[
x_{(k+1)} \textstyle \leftarrow \tilde x_{(h)}
\]
\rule{\textwidth}{.4pt}
\end{minipage}}
\vspace{0.3cm}
\item \textit{Dual variables are updated}
\begin{align}\label{eq:outer}
 \lambda_{(k+1)} & \leftarrow \Pi_{\mathbb{R}^{m}_{\ge0}}[\lambda_{(k)}-\tau (b-Ax_{(k+1)})] \\
 k & \leftarrow k+1.
\end{align}
\end{enumerate}}
\end{algorithm}

The inner loop in Algorithm \ref{alg:two} converges to a Wardrop equilibrium of the game $\mc{G}(\lambda_{(k)})$ under the following assumption.
\begin{assumption} \label{A4}
There exists $L>0$ such that, for all $i\in\{1,\dots,\N\}$ and $\lambda\in\R^m_{\ge0}$, the mapping $z\mapsto x^i_{\textup{or}}(z,\lambda)$ is single valued and Lipschitz with constant smaller than $L$. Moreover, at least one of the following statements holds.
\begin{enumerate}
\item For each $i\in\{1,\dots,\N\}$ and $\lambda \in \R_{\ge 0}^m$, the mapping $z\mapsto x^i_{\textup{or}}(z,\lambda)$ is non-expansive\footnote{The mapping is non-expansive if $\| x^i_{\textup{or}}(z_1,\lambda) - x^i_{\textup{or}}(z_2,\lambda) \| \le \| z_1 - z_2\|$ for all $z_1$, $z_2$.}.
\item For each $i\in\{1,\dots,\N\}$ and $\lambda \in \R_{\ge 0}^m$, the mapping $z\mapsto z-x^i_{\textup{or}}(z,\lambda)$ is strongly monotone. \qed
\end{enumerate}
\end{assumption}
\noindent Sufficient conditions for Assumption~\ref{A4} to hold are in~\cite[Corollary 1]{grammatico:parise:colombino:lygeros:14} for $v\i$, $p$ in~\eqref{eq:costs_specific} respectively quadratic and affine.
\begin{theorem}
\label{thm:conv_two}
Suppose that the operator $F\WE$ in~\eqref{eq:F_W} is strongly monotone on $\mc{X}$ with constant $\alpha$, that Assumptions~\ref{A1},~\ref{A3},~\ref{A4} hold, and that $\mc{X}\i$ is bounded for all $i \in \{1,\dots,\N\}$; set $\tau<\frac{2\alpha}{\|A\|^2}$ in~\eqref{eq:outer}. Then $x_{(k)}$ in Algorithm \ref{alg:two} converges to a variational Wardrop equilibrium of $\mc{G}$.
\myqed
\end{theorem}

\begin{remark}[Convergence rate]
The convergence rate of Algorithm \ref{alg:two} is  an open question. Nonetheless, it is possible to characterize the convergence rate in both of the two levels for some special cases. Specifically, under Assumption  \ref{A4}-2)  it is possible to modify  line \eqref{eq:inf_inner_c} with $z_{(h+1)} \leftarrow (1-\frac1\mu)z_{(h)}+\frac{1}{\mu}\tilde \sigma_{(h+1)}$ and guarantee geometric convergence for $\mu\in[0,1]$ small enough, see e.g.~\cite[Theorem 3.6 (iii)]{berinde}. The outer loop on the other hand has geometric convergence under the additional assumption that the mapping $\Phi$ as defined in the proof of Theorem \ref{thm:conv_two} is not only co-coercive but also strongly monotone.
\end{remark}
The proof is given in the Appendix.
To the best of our knowledge this is the first algorithm that guarantees convergence to a Wardrop equilibrium in games with coupling constraints by using \textit{optimal responses.} We note that, for the case of specific cost~\eqref{eq:costs_specific} and $p$ affine, \cite{Grammatico2016Arxiv} proposes a one-level optimal response algorithm that converges to a pair $(\bar x,\bar \lambda)$ such that $\bar x$ is a Wardrop equilibrium of the game $\mc{G}(\bar \lambda)$ satisfying the coupling constraint. However such point is not a Wardrop equilibrium because the complementarity condition $0\le \bar \lambda \perp b-A\bar x\ge 0$ is not guaranteed. A two-level \textit{gradient-step} algorithm for Nash equilibrium with coupling constraints has been proposed in~\cite[Algorithm 2]{pang2010design} and in~\cite[Section 4]{pavel2007extension}.
\subsection{Asymmetric projection algorithm based on gradient step for Nash and Wardrop equilibrium}
\label{subsec:boundedly}
We propose here an algorithm to achieve a Nash or a Wardrop equilibrium by making use of the equivalent reformulation of VI$(\mathcal Q, F)$ as the extended  VI$(\mathcal{Y},T)$ given in Proposition~\ref{prop:ext_vi}. Solving VI$(\mathcal{Y},T)$ instead of VI$(\mathcal Q, F)$ allows the design of a decentralized algorithm, because the set $\mathcal{Y}$ is the Cartesian product $\mc{X}^1 \times \dots \mc{X}^\N \times \R^m_{\ge 0}$, and thus the individual constraint sets $\mc{X}\i$ are decoupled.

Algorithm~\ref{alg:asp} finds a solution of VI$(\mathcal{Y},T)$, where $T$ is as in~\eqref{eq:T}, with $F = F\NE$, and hence achieves a Nash equilibrium. If the same algorithm is used with $F = F\WE$ it achieves a Wardrop equilibrium.
At every iteration each agent computes his new strategy $x^{i}_{(k+1)}$ by taking a gradient step, based on his previous strategy $x^{i}_{(k)}$, the previous average $\sigma(x_{(k)})$ and the previous dual variables $\lambda_{(k)}$. Given the new coupling constraint violation, the  central operator updates the price to $\lambda_{(k+1)}$ and broadcasts it to the agents.
\begin{algorithm}[H]
\caption{for Nash and Wardrop equilibria}
\label{alg:asp}
{\small
\textbf{Initialization}: Set $k =0$, $\tau>0$, $x^i_{(0)} \in\R^n$, $\lambda_{(0)}\in\R^m_{\ge0}$.\\
\textbf{Iterate until convergence}:
\begin{subequations}
\label{eq:apa_inner}
\begin{align}
\hspace{-0.7cm} \sigma_{(k)} &\textstyle \leftarrow \frac{1}{\N}\sum_{i=1}^\N x^{i }_{(k)} \label{eq:apa_inner_a}\\
\hspace{-0.7cm} x^{i}_{(k+1)}&\leftarrow \! \Pi_{\mathcal{X}^i}[x^{i }_{(k)} \! - \! \tau \Bigl(\nabla_{\!\! x\i} J\i(x\i_{(k)},\sigma(x_{(k)}\!)) \! + \! {A}_{(:,i)}^\top\lambda_{(k)} \Bigr)], \! \forall i   \label{eq:apa_inner_b}\\
\hspace{-0.7cm} \lambda_{(k+1)} & \leftarrow \! \Pi_{\mathbb{R}^{m}_{\ge0}}[\lambda_{(k)}-\tau (b-2Ax_{(k+1)}+Ax_{(k)})] \label{eq:apa_inner_c} \\
k & \leftarrow k+1.
\end{align}
\end{subequations}}
\vspace*{-0.3cm}
\end{algorithm}
\begin{theorem}
\label{thm:convergence_asp}
Let Assumptions~\ref{A1} and~\ref{A3} hold. Then
\begin{itemize}
\item Let $F\NE$ in~\eqref{eq:F_N} be strongly monotone on $\mc{X}$ with constant $\alpha$ and Lipschitz on $\mc{X}$ with constant $L_F$. Set $\tau>0$ s.t.
\begin{equation}
\textstyle \tau <\frac{-L_F^2+\sqrt{L_F^4+4\alpha^2\|A\|^2}}{2\alpha \|A\|^2 }\ . \label{eq:condition_tau_APA}
\end{equation}
Then $x_{(k)}$ in Algorithm~\ref{alg:asp} converges to a variational Nash equilibrium of $\mc{G}$ in~\eqref{eq:GNEP}.
\item Let $F\WE$ in~\eqref{eq:F_W} be strongly monotone on $\mc{X}$ with constant $\alpha$ and Lipschitz on $\mc{X}$ with constant $L_F$, then Algorithm~\ref{A2} with $\nabla_{x\i} J\i(x\i_{(k)},z)\eval$ in place of $\nabla_{x\i} J\i(x\i_{(k)},\sigma(x_{(k)}))$ converges to a variational Wardrop equilibrium, if $\tau$ satisfies~\eqref{eq:condition_tau_APA}. \hfill{$\square$}
\end{itemize}
\end{theorem}
\begin{remark}[Convergence rate]
By specializing the result in \cite{marcotte1995convergence} we proved in \cite[Proposition 1]{paccagnan2016distributed} that if the operator $F$ is not only monotone but also affine and the set $\mathcal{X}$ is a polyhedron  then for $\tau$ small enough Algorithm \ref{alg:asp} converges $R$-linearly, i.e., $\lim\sup_{k \rightarrow \infty} (\|y_{(k)}-\bar y\|)^{\frac1k}<1.$
\end{remark}

The proof is given in the Appendix and is based on the fact that Algorithm \ref{alg:asp} is a specific type of asymmetric projection algorithm \cite[Algorithm 12.5.1]{facchinei2007finite} applied to VI$(\mathcal{Y},T)$. A proof for the case in which $F$ is affine and symmetric is given in~\cite[Propositions 2 and 4]{tseng1990further}. 
We briefly note that there are other gradient based algorithms that can be implemented in a decentralized fashion to solve VI$(\mathcal{Y},T)$. One example is the extragradient algorithm \cite[Algorithm 12.1.9]{facchinei2007finite}. This would however  require two updates for both $x$ and  $\lambda$ at each iteration.

\subsection{Convergence guarantees for quadratic games}
In the previous subsections we have proposed two different algorithms. We summarize in Table \ref{tb:summary} the main conditions that guarantee their convergence.
\begin{table}[h]
\begin{center}
\begin{tabular}{ccc}\hline
& Nash & Wardrop\\ \hline
optimal response& \multirow{2}{*}{ -} &  $F\WE$ strongly monotone \\
(Algorithm \ref{alg:two})&  &  and Assumption \ref{A4} \\ \hline
gradient step & \multirow{2}{*}{$F\NE$ strongly monotone} & \multirow{2}{*}{$F\WE$ strongly monotone} \\
(Algorithm \ref{alg:asp})& & \\ \hline
\end{tabular}
\end{center}
\caption{Range of applicability of the presented algorithms, under Assumptions~\ref{A1} and~\ref{A2}.}
\label{tb:summary}
\end{table}%
To better understand the differences and the range of applicability of the two algorithms we refine the sufficient conditions of Table~\ref{tb:summary} to the important class of aggregative games with quadratic cost functions
\begin{equation}
J\i(x\i,\sigma(x)) \defeq \frac{1}{2} (x\i)^\top Q x\i +(C\sigma(x)+c\i)^\top x\i\,,
\label{eq:costs_quad}
\end{equation}
where $Q \in \R^{n\times n}$ is symmetric, $C\in\R^{n\times n}, c\i \in\R^n$.
These cost functions have been used in~\cite{huang2007large,grammatico:parise:colombino:lygeros:14,bauso:pesenti:13}. Since the operators $F\NE,F\WE$ defined in \eqref{eq:F} are obtained by differentiating quadratic functions, their expression is given by
\begin{subequations}
\label{eq:F_quad}
\begin{align}
\label{eq:F_quad_W}
F\WE(x)&= \textstyle \left(I_\N\otimes Q + \frac 1 \N \mathbbold{1}_\N\mathbbold{1}_\N^\top \otimes C \right) x+  c, \\
F\NE(x)&= \textstyle F\WE(x)+  \frac{1}{\N} (I_\N \otimes C^\top)  x,
\label{eq:F_quad_N}
\end{align}
\end{subequations}
where $c=[c^1;\ldots; c^\N]$.
The following lemma exploits the characterization~\eqref{eq:F_quad} to derive sufficient conditions for strong monotonicity of $F\WE$, $F\NE$ and for Assumption~\ref{A4}. These in turn guarantee convergence of Algorithm~\ref{alg:two} and~\ref{alg:asp} as by Table~\ref{tb:summary}.
\begin{lemma}
\label{lem:quadratic}
The following hold.
\begin{itemize}
\item If $Q\succ 0$, $C\succeq 0$ or if $Q\succeq 0$, $C\succ 0$ then $F\NE$ in~\eqref{eq:F_quad_N} is strongly monotone.
\item If $Q\succ 0$, $C\succeq 0$ then $F\WE$ in~\eqref{eq:F_quad_W} is strongly monotone.
\item If $Q\succ 0$, $C=C^\top \succ 0$ or if $Q\succ 0$, $Q - C^\top Q^{-1} C \succ 0$ then $F\WE$ in~\eqref{eq:F_quad_W} is strongly monotone and Assumption~\ref{A4} is satisfied.
\qed
\end{itemize}
\end{lemma}
\begin{proof}
By Lemma~\ref{lemma:pd}, strong monotonicity of $F\WE$ in~\eqref{eq:F_quad_W} is equivalent to $\nabla_x F\WE(x) = \left(I_\N\otimes Q + \frac{1}{\N} \mathbbold{1}_\N\mathbbold{1}_\N^\top \otimes C \right)^\top \succ 0$, which is independent from $x$. Similarly, strong monotonicity of $F\NE$ in~\eqref{eq:F_quad_N} is equivalent to $\left(I_\N\otimes Q + \frac{1}{\N} \mathbbold{1}_\N\mathbbold{1}_\N^\top \otimes C \right)^\top +  \frac{1}{\N} (I_\N \otimes C^\top)^\top \succ 0$. Building on this, the first two statements are straightforward to prove. Regarding the last statement, $Q\succ 0$, $C=C^\top \succ 0$ imply $\nabla_x F\WE(x) \succ 0$. Moreover, by~\cite[Theorem 2]{grammatico:parise:colombino:lygeros:14}, Assumption~\ref{A4}.2 is satisfied. By using Schur's theorem, it can be shown that $Q \succ 0$, $Q - C^\top Q^{-1} C \succ 0$ imply $Q+C \succ 0$, hence $\nabla_x F\WE(x) \succ 0$. Finally, by~\cite[Theorem 2]{grammatico:parise:colombino:lygeros:14}, Assumption~\ref{A4}.1 is satisfied.
\end{proof}
\section{Charging of electric vehicles}
\label{sec:PEVs}

We model the simultaneous charging of a population of electric vehicles (EV) as a game, following the approach of \cite{ma2013decentralized,grammatico:parise:colombino:lygeros:14,dario2015aggregative}. Compared to the existing work, our main contributions consist in introducing the coupling constraints, finding a Nash and a Wardrop equilibrium even for the case of $v\i = 0$ in~\eqref{eq:costs_specific}, and studying the distance between the aggregate strategies at the Nash and at the Wardrop equilibrium.
\subsection*{Constraints}
We consider a population of $\N$ electric vehicles. The state of charge of vehicle $i$ at time $t$ is described by the variable $s\i_t$. The time evolution of $s^i_t$ is specified by the discrete-time system $s\i_{t+1} = s\i_t + b\i x\i_t \,,  t = 1, \dots, n$, where $x\i_t$ is the charging control and the parameter $b\i > 0$ is the charging efficiency. We assume that the charging control cannot take negative values and that at time $t$ it cannot exceed $\tilde x^i_t \ge 0$. The final state of charge is constrained to $s_{n+1}^i\ge\eta\i$, where $\eta\i \ge 0$ is the desired state of charge of agent $i$. Denoting $x\i =[x\i_1, \dots, x^i_n]^\top \in \R^n$, the individual constraint of agent $i$ can be expressed as
\begin{equation}
\label{eq:vehicle_constraint}
x\i \in \mc{X}\i \vcentcolon=\!\! \left\{ x\i \in \mathbb{R}^n  \left|
\begin{array}{l}
0 \le x\i_t \le \tilde x\i_t, \;\; \forall \, t=1,\dots,n \\ 
\sum_{t=1}^{n} x\i_t \ge \theta\i
\end{array}
\!\!\!\!\right.
\right\},
\end{equation}
where $\theta\i \coloneqq {(b\i)}^{-1} (\eta\i - s\i_1)$, with $s\i_1 \ge 0$ the state of charge at the beginning of the time horizon. Besides the individual constraints $x\i \in \mc{X}\i$, we also introduce the coupling constraint
\begin{equation}
x\in\mc{C}\defeq\{x\in\R^{\N n}\mid \textstyle \frac{1}{\N}\sum_{i=1}^\N x\i_t \le K_t,  \, \forall \, t=1,\dots,n\},
\label{eq:coupling_global_PEVs}
\end{equation}
indicating that  at time $t$ the grid cannot deliver more than $\N \cdot K_t$ units of power to the vehicles. In compact form~\eqref{eq:coupling_global_PEVs} reads as
$(\ones[\N]^\top \otimes I_n) x \le\N K \,,$
where $K \defeq [K_1, \dots, K_n]^\top$.
\subsection*{Cost function}
The cost function of each vehicle represents its electricity bill, which we model as
\begin{equation}
\textstyle J\i(x\i,\sigma(x))=\sum_{t=1}^n p_t \left( \frac{d_t + \sigma_t(x)}{\kappa_t}  \right) x\i_t \eqdef p(\sigma(x))^\top x^i,
\label{eq:PEV_energy_bill}
\end{equation}
where we assumed that the energy price for each time interval $p_t:\R_{\ge0}\rightarrow \R_{>0}$ depends on the ratio between total consumption and total capacity $(d_t + \sigma_t(x))/ \kappa_t$, where $d_t$ and $\sigma_t(x)\defeq\frac{1}{\N}\sum_{i=1}^\N x^i_t$ are the non-EV and EV demand at time $t$ divided by $\N$ and $\kappa_t$ is the total production capacity divided by $\N$ as in~\cite[eq. (6)]{ma2013decentralized}. $\kappa_t$ is in general not related to $K_t$.
\subsection{Theoretical guarantees}
We  define the  game $\mc{G}^\text{EV}_\N$ as in~\eqref{eq:GNEP}, with $\mathcal{X}\i$, $\mathcal{C}$ and $J\i(x\i,\sigma(x))$ as in \eqref{eq:vehicle_constraint}, \eqref{eq:coupling_global_PEVs} and \eqref{eq:PEV_energy_bill} respectively.
In the following corollary we refine the main results of Sections~\ref{sec:connection_VI},~\ref{sec:Wardrop_Nash},~\ref{sec:algorithms} for the EV application.
\begin{corollary}
\label{cor:pev} 
Consider a sequence of games $(\mc{G}^\textup{EV}_M)_{M=1}^\infty$. Assume that there exists $\tilde x^0$ such that $\tilde x^i_t \le \tilde x^0$ for all $t \in \{1,\dots,n\},i \in \{1,\dots,\N\}$ and for each game $\mc{G}^\textup{EV}_M$. Moreover, assume that for each game $\mc{G}^\textup{EV}_M$ the set $\mc{Q}=\mc{C}\cap\mc{X}$ is non-empty and that for each $t$ the price function $p_t$ in \eqref{eq:PEV_energy_bill} is twice continuously differentiable, strictly increasing and Lipschitz in $[0,\tilde x^0]$ with constant $L_p$.
Moreover, assume
\begin{equation}
\minn{\substack{t \in \{1,\dots,n\}\\z \in [0,\tilde x^0]}}{\left(p'_t(z) - \frac{\tilde x^0 p''_t(z)}{8}\right)} > 0.
\label{eq:bound_PEV_smon}
\end{equation}
Then:
\begin{enumerate}
\item
A Wardrop and a Nash equilibrium exist for each game $\mathcal{G}^\textup{EV}_M$ of the sequence. Furthermore, every Wardrop equilibrium is an $\varepsilon$-Nash equilibrium with $\varepsilon = \frac{2n (\tilde{x}^0)^2 L_p}{\N}$.
\item The function $p$ is strongly monotone, hence for each game $\mathcal{G}^\textup{EV}_M$ there exists a unique $\bar \sigma$ such that $\sigma(\VWE{x})=\bar \sigma$ for any variational Wardrop equilibrium $\VWE{x}$ of $\mathcal{G}^\textup{EV}_M$. Moreover for any variational Nash equilibrium $\VNE{x}$ of $\mathcal{G}^\textup{EV}_M$, $\| \sigma(\VNE{x}) - \sigma(\VWE{x}) \| \le \tilde x^0 \sqrt{\frac{2 n L_p}{\alpha M}}$, where $\alpha$ is the monotonicity constant of $p$.
\item For each game $\mathcal{G}^\textup{EV}_\N$ the operator $F\WE$ is monotone, hence the extragradient algorithm \cite[Algorithm 12.1.9]{facchinei2007finite} with operator $F\WE$ converges to a variational Wardrop equilibrium of $\mathcal{G}^\textup{EV}_M$. 
\item For each game $\mathcal{G}^\textup{EV}_\N$ the operator $F\NE$ is strongly monotone. Hence, Algorithm~\ref{alg:asp} converges to a variational Nash equilibrium of $\mathcal{G}^\textup{EV}_M$. 
\qed
\end{enumerate}
\end{corollary}
\begin{proof}
1) We show that Assumption~\ref{A1} holds. Indeed the sets $\mc{X}^i$ in \eqref{eq:vehicle_constraint} are convex and compact, the function $g$ in~\eqref{eq:coupling_constraints_general} is affine and hence convex, and $\mc{Q}$ is non-empty by assumption.
For each $z$ fixed, the function $J^i(x^i,z)$ is linear hence convex in $x^i$. We prove in the last statement that $F\NE$ is strongly monotone. This is equivalent to $\nabla_x F\NE(x) \succ 0$ by Lemma~\ref{lemma:pd}, which by definition of $F\NE(x)$ implies $\nabla_{x\i} (\nabla_{x\i} J\i(x\i,\sigma(x))) \succ 0$, which implies convexity of $J\i(x\i,\sigma(x))$. Finally, $J^i(z_1,z_2)$ is continuously differentiable in $[z_1;z_2]$ because $p_t$ is twice continuously differentiable.
Having verified Assumption~\ref{A1}, Lemma \ref{lem:exun} guarantees the existence of a Nash and of a Wardrop equilibrium. The $\varepsilon$-Nash property is guaranteed by Proposition \ref{prop:conv_cost} upon verifying Assumption~\ref{A2}. This holds because: i) $\cup_{i=1}^\N \mathcal{X}^i\subseteq{ [0,\tilde x^0]^n}$, ii)  $J^i(z_1,z_2)$ is Lipschitz in $z_2$ on $[0,\tilde x^0]^n$ with Lipschitz constant $L_2 = R L_p$, iii) \eqref{eq:Lipschitz_implies_Lipschitz} holds and iv) $p_t$ is assumed Lipschitz in $[0,\tilde x^0]$ with Lipschitz constant $L_p$ for all $t$. We conclude by noting that $R = \tilde x^0 \sqrt{n}$.
\\
2) The fact that each $p_t$ is strictly increasing in $[0,\tilde x^0]$ implies that $\nabla_z p(z) \succ 0$ in $[0,\tilde x^0]^n$, where $p(z)\defeq\left[p_1(\frac{d_1 + z_1}{\kappa}),\ldots,p_n(\frac{d_n + z_n}{\kappa})\right]^\top$. In turn $\nabla_z p(z) \succ 0$ guarantees strong monotonicity of $p$ in $[0,\tilde x^0]^n$ by Lemma~\ref{lemma:pd}. This, together with Assumptions \ref{A1} and \ref{A2} verified above, allows us to use the third result in Theorem \ref{thm:conv_strategies}.
\\
3) Since $\mathcal X$ is closed and convex,~\cite[Theorem 12.1.11]{facchinei2007finite} guarantees that the extragradient algorithm converges to a Wardrop equilibrium if $F\WE$ is monotone, which follows from the first statement of Lemma~\ref{lem:FNstrongly monotone}.
\\
4)
Assumption 1, which has been shown to hold in the first statement, and Assumption 3, which trivially holds, allow us to use Theorem~\ref{thm:convergence_asp}, upon showing strong monotonicity of $F\NE$. We have proven in the third statement that $F\WE$ is monotone. According to~\eqref{eq:F_N_decomp}, to show strong monotonicity of $F\NE$ it is sufficient to show that under condition~\eqref{eq:bound_PEV_smon} the term $[\nabla_z p(z)_{|z=\sigma(x)} {x\i}  ]_{i=1}^\N$ is strongly monotone for all $x \in \mc{X}$, which is equivalent to $\nabla_x [\nabla_z p(z)_{|z=\sigma(x)} {x\i}  ]_{i=1}^\N \succ 0$  for all $x \in \mc{X}$ by Lemma~\ref{lemma:pd}.
We have
\begin{equation}
\begin{aligned}
&\nabla_x [\nabla_z p(z)_{|z=\sigma(x)} x\i]_{i=1}^\N = \\
&I_\N \otimes \nabla_z p(z)_{|z=\sigma(x)} + \frac{1}{\N} \ones[\N] \otimes \left([ \text{diag} \{p''_t(\sigma_t) x^i_t \}_{t=1}^n ]_{i=1}^{\N}\right)^\top \!\!\!,
\end{aligned}
\label{eq:PEV_proof_intermediate}
\end{equation}
where $\text{diag} \{p''_t(\sigma_t) x^i_t \}_{t=1}^n$ is the diagonal matrix whose entry in position $(t,t)$ is $p''_t(\sigma_t) x^i_t$. The permutation matrix $P = [[e_{t+(i-1)n}^\top]_{i=1}^\N]_{t=1}^n$ permutes~\eqref{eq:PEV_proof_intermediate} into block-diagonal form

\begin{align}
&P \nabla_x [\nabla_z p(z)_{|z=\sigma(x)} x\i]_{i=1}^\N P^\top = \label{eq:PEV_gradient_FN} \\
&\begin{small} \begin{bmatrix}
p_1'(\sigma_1) I_{\! \N} \!\!\! & & \\
& \!\!\!\! \ddots \!\!\!\! & \\
& & \!\! p_n'(\sigma_n) I_{\! \N} \!
\end{bmatrix}
\!\! + \! \frac{1}{\N} \!\!	
\begin{bmatrix}
p_1''(\sigma_1) x_1 \onestight[\N]^\top \!\! & & \\
& \!\!\!\! \ddots \!\!\!\!\! & \\
& & \!\!\! p_n''(\sigma_n) x_n \onestight[\N]^\top
\end{bmatrix} \end{small}
\end{align}
\noindent where $x_t  =  [x\i_t]_{i=1}^\N$. It suffices to show $p_t'(\sigma_t) I_\N + \frac{1}{\N} p''_t(\sigma_t) x_t \ones[\N]^\top \succ 0$ for all $t$.
By Lemma~\ref{lem:min_eigenval} in Appendix, $\lambda_\text{min} \left( x_t \ones[\N]^\top + \ones[\N] x_t^\top \right)/2 \ge -\frac{\tilde x^0 \N}{8}$, which ends the proof\footnote{The work~\cite{yin2011nash} studies an aggregative game and in~\cite[Lemma 3]{yin2011nash} it exploits expression~\eqref{eq:PEV_gradient_FN} to give conditions for $\nabla_x F\NE(x)$ to be a $P$-matrix, which in turn guarantees uniqueness of the Nash equilibrium in absence of coupling constraints. It is interesting to note that uniqueness in~\cite{yin2011nash} holds assuming $p'_t > 0, p''_t > 0$, whereas for us it suffices $p'_t > 0, p''_t < 0$.}.
\end{proof}
The average population strategy plays an important role in the EV application: indeed,~\cite[Theorem 6.1]{ma2013decentralized} shows in the same game setup that the average population strategy relative to a Nash equilibrium presents desirable properties for the grid operator. Nonetheless, if condition~\eqref{eq:bound_PEV_smon} is not satisfied, a Nash equilibrium cannot be achieved; it is instead possible to achieve a Wardrop equilibrium with the extragradient algorithm.
The second statement of Corollary~\ref{cor:pev} then provides guarantees on the distance between the average population strategies at the Nash and at the Wardrop equilibrium.
\subsection*{Uniqueness of dual variables.}
Corollary~\ref{cor:pev} shows that under condition~\eqref{eq:bound_PEV_smon} the operator $F\NE$ of $\mathcal{G}^\textup{EV}_\N$ is strongly monotone, hence the game $\mathcal{G}^\textup{EV}_\N$ admits a unique variational Nash equilibrium (Lemma~\ref{lem:exun}). We study here the uniqueness of the associated dual variables $\bar \lambda\NE$ introduced in Proposition~\ref{prop:ext_vi}. Guaranteeing unique dual variables might be important to convince the vehicle owners to participate in the proposed scheme, as the dual variables represent the penalty price associated to the coupling constraint. 
Define $R^\text{tight} \subseteq \{1,\dots,n\}$ as the set of instants in which the coupling constraint $\mc{C}$ is active. We provide a sufficient condition for uniqueness of the dual variables which relies on a modification of the linear-independence constraint qualification~\cite{Wachsmuth2013}.
\\
\begin{proposition}
\label{prop:uniqueness_for_PEVs}
Assume that condition~\eqref{eq:bound_PEV_smon} holds and consider the unique variational Nash equilibrium $\VNE{x}$ of $\mc{G}^\textup{EV}_M$. If there exists a vehicle $i$ such that 
\begin{itemize}
\item $\bar x_{\textup{N},t}^i \notin \{0,\tilde x\i_t \} \quad \text{for all} \: t \in R^\textup{tight}$ and
\item $\bar x_{\textup{N},t'}^i \notin \{0,\tilde x\i_{t'} \} \quad \text{for some} \: t'\notin R^\textup{tight}$,
\end{itemize}
then the dual variables $\VNE{\lambda}$ associated to the coupling constraint~\eqref{eq:coupling_global_PEVs} are unique. 
\qed
\end{proposition}
The proof is reported in the Appendix.
We note that the sufficient condition of Proposition~\ref{prop:uniqueness_for_PEVs} is to be verified a-posteriori; in other words, it depends on the primal solution $\VNE{x}$. In the numerical analysis presented in the following such sufficient condition always holds. Uniqueness of the dual variables associated to the coupling constraint of an aggregative game has been studied also in~\cite[Theorem 4]{yin2011nash}, where the conditions in the bullets of Proposition~\ref{prop:uniqueness_for_PEVs} are not required but $p$ is restricted to be affine.
\subsection{Numerical analysis}\label{sec:pev_num}
The numerical study is conducted on a heterogeneous population of agents. We set the price function to  $p_t(z_t)=0.15 \sqrt{(d_t+\sigma_t(x)) / \kappa_t}$ and $n=24$. The agents differ in $\theta\i$, randomly chosen according to $\mathcal{U}[0.5,1.5]$; they also differ in $\tilde x\i_t$, which is chosen such that the charge is allowed in a connected interval, with left and right endpoints uniformly randomly chosen: within the interval, $\tilde x\i_t$ is constant and randomly chosen for each agent, according to $\mathcal{U}[1,5]$; outside this interval, $\tilde x\i_t = 0$.
The demand $d_t$ is taken as the typical (non-EV) base demand over a summer day in the United States~\cite[Figure 1]{ma2013decentralized}; $\kappa_t=12$ kW for all $t$, and the upper bound $K_t=0.55$ kW is chosen such that the coupling constraint~\eqref{eq:coupling_global_PEVs} is active in the middle of the night.
Note that with these choices all the assumptions of Corollary~\ref{cor:pev} are met. In particular, for the given choice of $p$ condition~\eqref{eq:bound_PEV_smon} holds because $p''_t(z) < 0$ for all $z$ and all $t$. Figure~\ref{fig:primal_variables} presents the aggregate consumption at the Nash equilibrium found by Algorithm~\ref{alg:asp}, with stopping criterion $\|(x_{(k+1)},\lambda_{(k+1)})-(x_{(k)},\lambda_{(k)})\|_{\infty} \le 10^{-4}$.
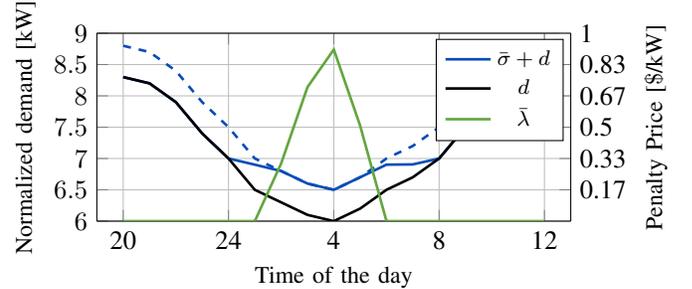
\begin{figure} [h!] 
\begin{center}
\newlength\figureheight 
\newlength\figurewidth 
\setlength\figureheight{2.5cm} 
\setlength\figurewidth{6.3cm} 
%
%
\definecolor{mycolor1}{rgb}{0.03,0.3,0.72}%
\definecolor{PortlandGreen}{RGB}{99,166,63}%
\begin{tikzpicture}

\begin{axis}[%
width=\figurewidth,
height=\figureheight,
axis y line* = left,
scale only axis,
xmin=7, xmax=25, xtick={8,12,16,20,24}, xticklabels={{20},{24},{4},{8},{12}}, xlabel={\small Time of the day}, xmajorgrids,
ymin=6, ymax=9, ytick={6,6.5,7,7.5,8,8.5,9}, ylabel={\small Normalized demand [kW]}, ylabel near ticks, ymajorgrids,
]
\addplot [color=mycolor1,solid,line width=1.0pt]
  table[row sep=crcr]{%
8	8.30000000390642\\
9	8.20000000478831\\
10	7.9000000074229\\
11	7.40000002443321\\
12	7.00003584722159\\
13	6.90146739363718\\
14	6.80006156578355\\
15	6.60004473212242\\
16	6.50001520367497\\
17	6.70001392123971\\
18	6.90188130551423\\
19	6.90551484838831\\
20	7.0000246029291\\
21	7.50000000931836\\
22	7.60000000381587\\
23	7.80000000115141\\
24	8.1\\
};\label{plot_one}

\addplot [color=black,solid,line width=1.0pt]
  table[row sep=crcr]{%
8	8.3\\
9	8.2\\
10	7.9\\
11	7.4\\
12	7\\
13	6.5\\
14	6.3\\
15	6.1\\
16	6\\
17	6.2\\
18	6.5\\
19	6.7\\
20	7\\
21	7.5\\
22	7.6\\
23	7.8\\
24	8.1\\
};\label{plot_two}

\addplot [color=mycolor1,dashed,line width=1.0pt]
  table[row sep=crcr]{%
8	8.8\\
9	8.7\\
10	8.4\\
11	7.9\\
12	7.5\\
13	7\\
14	6.8\\
15	6.6\\
16	6.5\\
17	6.7\\
18	7\\
19	7.2\\
20	7.5\\
21	8\\
22	8.1\\
23	8.3\\
24	8.6\\
};
\end{axis}

\begin{axis}[
width=\figurewidth,
height=\figureheight,
axis y line*=right,
ymin = 0, ymax = 1, ytick={0.16666666,0.3333333,0.5,0.666666,0.83333333,1}, yticklabels={0.17,0.33,0.5,0.67,0.83,1}, ylabel={\small Penalty Price [${\$}$/{kW}]}, ylabel near ticks,
xmin=7, xmax=25, axis x line=none,
scale only axis,
legend style={at={(0.985,0.97)},font=\footnotesize}
     ]
\addplot [color=PortlandGreen,solid,line width=1.0pt,forget plot]
  table[row sep=crcr]{%
8	0\\
9	0\\
10	0\\
11	0\\
12	0\\
13	0\\
14	0.306979585957893\\
15	0.713221133089155\\
16	0.913221133089153\\
17	0.508096714131238\\
18	0\\
19	0\\
20	0\\
21	0\\
22	0\\
23	0\\
24	0\\
};\label{plot_three}
\addlegendimage{/pgfplots/refstyle=plot_one}\addlegendentry{$\bar{\sigma}+d$}
\addlegendimage{/pgfplots/refstyle=plot_two}\addlegendentry{$d$}
\addlegendimage{/pgfplots/refstyle=plot_three}\addlegendentry{$\bar{\lambda}$}
\end{axis}

\end{tikzpicture}%
\caption{\small Aggregate EV demand $\sigma(\VNE{x})$ and dual variables $\bar \lambda\NE$ for $\N=100$, subject to $\sigma(x) \le 0.55$ kW. The region below the dashed line corresponds to $\sigma(x)+d \le  0.55$ kW$+d$.}
\label{fig:primal_variables}
\end{center}
\end{figure}
Note that without the coupling constraint the quantity $\bar \sigma + d$ would be constant overnight, as shown in~\cite{ma2013decentralized}. Figure \ref{fig:distance_aggregates} illustrates the bound $\| \sigma(\VNE{x}) - \sigma(\VWE{x}) \| \le \tilde x^0 \sqrt{\frac{2 n L_p}{\alpha M}}$ of the second statement of Corollary~\ref{cor:pev}. The Wardrop equilibrium is computed with the extragradient algorithm with stopping criterion $\|(x_{(k+1)},\lambda_{(k+1)})-(x_{(k)},\lambda_{(k)})\|_{\infty} \le 10^{-4}$. The $\varepsilon$-Nash property of the Wardrop equilibrium in Proposition~\ref{prop:conv_cost} can also be illustrated; a plot is omitted here for reasons of space.
\begin{figure} [h]
\begin{center}
\setlength\figureheight{2.0cm} 
\setlength\figurewidth{7.cm} 
%
%
\definecolor{mycolor1}{rgb}{0.00000,0.54118,0.90196}%
\begin{tikzpicture}

\begin{axis}[%
width=\figurewidth,
height=\figureheight,
scale only axis,
xmin=0,
xmax=800,
ymin=0,
ymax=0.15,
xtick={0,100,200,300, 400, 500, 600, 700, 800},
tick label style={/pgf/number format/fixed},
ytick={0,0.03,0.06,0.09,0.12,0.15},
xmajorgrids,
ymajorgrids,
xlabel ={Population size $M$},
legend style={at={(0.97,0.95)},legend cell align=left,align=left,draw=white!15!black,font=\footnotesize}
]
\definecolor{mycolor1}{rgb}{0.00000,0.54118,0.90196}

\addplot [color=black,solid,line width=1.0pt]
  table[row sep=crcr]{%
50	0.0961423497667815\\
100	0.0549965612736588\\
150	0.0368506226735783\\
200	0.0281709955720355\\
250	0.0233737077229369\\
300	0.0218642951846739\\
350	0.0180020163292226\\
400	0.0162783969936725\\
450	0.01454678132398\\
500	0.0119194105921108\\
550	0.0115542608715654\\
600	0.0101954546561331\\
650	0.0101114713636592\\
700	0.00906486857544332\\
750	0.00816693178529436\\
800	0.00792344623874586\\
};
\addlegendentry{\footnotesize $\|\sigma(\bar x_{_{N}})-\sigma(\bar x_{_{W}})\|$}

\addplot [dashed,line width=1.0pt]
  table[row sep=crcr]{%
50	0.14142135623731\\
100	0.1\\
150	0.0816496580927726\\
200	0.0707106781186548\\
250	0.0632455532033676\\
300	0.0577350269189626\\
350	0.0534522483824849\\
400	0.05\\
450	0.0471404520791032\\
500	0.0447213595499958\\
550	0.0426401432711221\\
600	0.0408248290463863\\
650	0.0392232270276368\\
700	0.0377964473009227\\
750	0.0365148371670111\\
800	0.0353553390593274\\
};
\addlegendentry{\footnotesize $1/\sqrt{M}$}

\end{axis}
\end{tikzpicture}%
\vspace{-0.1cm}
\caption{\small Distance between the aggregates $\sigma(\VNE{x})$ and  $\sigma(\VWE{x})$ at the Nash and Wardrop equilibrium (solid line). Corollary~\ref{cor:pev} ensure that such distance is upper bounded by  $K/\sqrt{M}$ for $K=\tilde x^0 \sqrt{2 n L_p/\alpha}$. The dotted line shows $1/\sqrt{M}$  proving that our bound has the right trend, while the constant  $K$ is, in this case, conservative.}
\label{fig:distance_aggregates}
\end{center}
\end{figure}
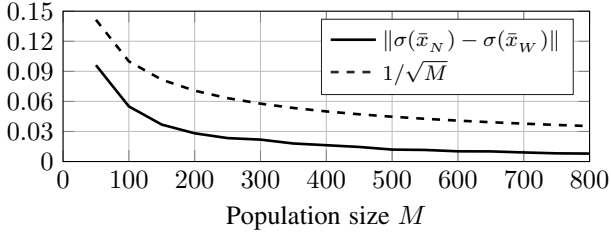
The framework introduced above can also be used to enforce local coupling constraints, i.e. constraints on a subset of all the vehicles. These can for instance be used to model capacity limits for local substations. We refer the reader to~\cite[Section VI]{paccagnan2016distributed} for a more detailed analysis.
\subsection*{Quadratic cost function}
Different works in the EV literature~\cite{grammatico:parise:colombino:lygeros:14,kristoffersen2011optimal} use the quadratic cost~\eqref{eq:costs_quad}, with $Q \succ 0$ and $C\succ 0$, diagonal. Existence of a Nash and of a Wardrop equilibrium is guaranteed by Lemma~\ref{lem:exun}, while Proposition~\ref{prop:conv_cost} gives the $\varepsilon$-Nash property.  Further, Lemma \ref{lem:quadratic} shows that the resulting operators $F\NE$ and $F\WE$ are strongly monotone with monotonicity constant independent from $\N$. Theorem \ref{thm:conv_strategies} ensures then that $\| \VNE{x} - \VWE{x} \| \le L_2/(\alpha \sqrt{\N})$, with $L_2 = R \cdot \lambda_\textup{max}(C)$.
A Nash equilibrium can be found using Algorithm \ref{alg:asp}, while a Wardrop equilibrium can be achieved using both Algorithm \ref{alg:two} and \ref{alg:asp}. Figure~\ref{fig:iterations_x} presents a comparison between the two algorithms in terms of iteration count, where $Q=0.1 I_n$, $C=I_n$, $c^i=d\; \text{for all} \; i$. Figure~\ref{fig:iterations_x} (top) represents the number of strategy updates required to converge, i.e. the number of times  \eqref{eq:inf_inner} or \eqref{eq:apa_inner_b} is used. Figure~\ref{fig:iterations_x} (bottom) depicts the number of dual variables updates, i.e. the number of times \eqref{eq:outer} or \eqref{eq:apa_inner_c} is used. For both algorithms the number of iterations does not seem to increase  with the population size. Algorithm  \ref{alg:asp} requires fewer primal iterations, while Algorithm~\ref{alg:two} needs much fewer dual iterations.
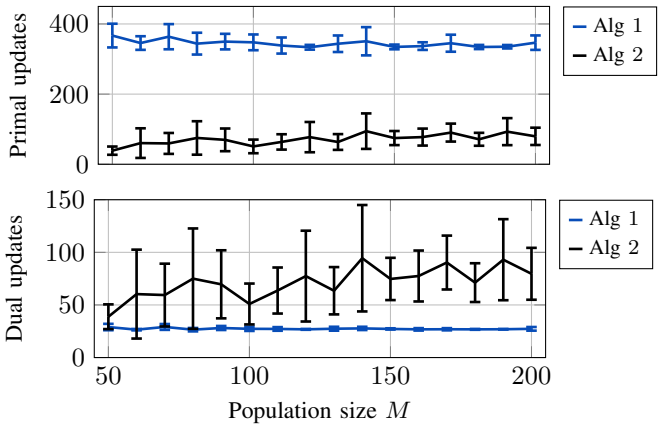
\begin{figure} [h!]
\begin{center}
\setlength\figureheight{2.1cm} 
\setlength\figurewidth{6cm}  
%
%
\begin{tikzpicture}
\definecolor{mycolor1}{rgb}{0.03,0.3,0.72}%
\definecolor{mycolor2}{rgb}{0.93333,0.69804,0.00000}%
\begin{axis}[%
width=\figurewidth,
height=\figureheight,
scale only axis,
xmin=45,
xmax=205,
xmajorgrids,
ymin=0,
ymax=450,
ymajorgrids,
xticklabels={,,},
ylabel = {\small Primal updates},
legend pos = outer north east,
legend image post style={scale=0.3},
]
\addplot[color=mycolor1,line width=1.0pt]
 plot [error bars/.cd, y dir = both, y explicit, error bar style={line width=1pt},   error mark options={
      rotate=90,
      line width=1pt}]
 table[row sep=crcr, y error plus index=2, y error minus index=3]{%
50	366.9	33.8214429023955	33.8214429023955\\
60	345.6	19.4020617461135	19.4020617461135\\
70	363.7	35.8637700193385	35.8637700193385\\
80	343.9	31.1077160845987	31.1077160845987\\
90	349.7	22.1858062733812	22.1858062733812\\
100	347.6	22.4686448189471	22.4686448189471\\
110	338.5	22.9314194937863	22.9314194937863\\
120	333.6	6.78527818147495	6.78527818147495\\
130	343.6	23.5592869162036	23.5592869162036\\
140	350.7	40.2369233416274	40.2369233416274\\
150	334.5	6.72681202353685	6.72681202353685\\
160	336.9	10.8761206319165	10.8761206319165\\
170	345.1	24.2670558576849	24.2670558576849\\
180	334.3	5.93380147965872	5.93380147965872\\
190	335.1	4.92848861214064	4.92848861214064\\
200	346.7	20.6060670677352	20.6060670677352\\
};
\addlegendentry{\footnotesize Alg $1$};

\addplot [solid,line width=1.0pt]
 plot [error bars/.cd, y dir = both, y explicit, error bar style={line width=1pt},   error mark options={
      rotate=90,
      line width=1pt}]
 table[row sep=crcr, y error plus index=2, y error minus index=3]{%
50	38.8	11.7456374880208	11.7456374880208\\
60	60.3	42.2470117286418	42.2470117286418\\
70	59.4	29.8167738026769	29.8167738026769\\
80	75	47.6969600708473	47.6969600708473\\
90	69.6	32.3425416440947	32.3425416440947\\
100	50.9	19.423954283307	19.423954283307\\
110	63.7	21.8771570365073	21.8771570365073\\
120	77.4	43.1490440218552	43.1490440218552\\
130	63.5	22.477766792989	22.477766792989\\
140	94.4	50.5711380136931	50.5711380136931\\
150	74.7	20.0800896412342	20.0800896412342\\
160	77.5	24.2002066106883	24.2002066106883\\
170	90.3	25.6009765438743	25.6009765438743\\
180	71.2	18.4325798519903	18.4325798519903\\
190	93	38.5201246103903	38.5201246103903\\
200	79.6	24.654411369976	24.654411369976\\
};
\addlegendentry{\footnotesize{Alg $2$}};
\end{axis}
\end{tikzpicture}%
%
%
\begin{tikzpicture}
\definecolor{mycolor1}{rgb}{0.03,0.3,0.72}%
\definecolor{mycolor2}{rgb}{0.93333,0.69804,0.00000}%
\begin{axis}[%
width=\figurewidth,
height=\figureheight,
scale only axis,
xmin=45,
xmax=205,
xmajorgrids,
ymin=0,
ymax=150,
ymajorgrids,
ylabel = {\small Dual updates},
xlabel = {\small Population size $M$},
legend pos = outer north east,
legend image post style={scale=0.3},
]
\addplot [color=mycolor1,solid,line width=1.0pt]
 plot [error bars/.cd, y dir = both, y explicit, error bar style={line width=1pt},   error mark options={
      rotate=90,
      line width=1pt}]
 table[row sep=crcr, y error plus index=2, y error minus index=3]{%
50	29	3.03315017762062	3.03315017762062\\
60	26.5	0.670820393249937	0.670820393249937\\
70	29.1	2.77308492477241	2.77308492477241\\
80	26.5	1.62788205960997	1.62788205960997\\
90	28.1	1.92093727122985	1.92093727122985\\
100	27.3	2.00249843945008	2.00249843945008\\
110	27.2	1.72046505340853	1.72046505340853\\
120	26.8	0.4	0.4\\
130	27.4	1.85472369909914	1.85472369909914\\
140	27.7	1.552417469626	1.552417469626\\
150	27.2	0.6	0.6\\
160	26.8	1.16619037896906	1.16619037896906\\
170	26.9	1.04403065089105	1.04403065089105\\
180	26.8	0.4	0.4\\
190	26.9	0.3	0.3\\
200	27.3	1.73493515728975	1.73493515728975\\
};
\addlegendentry{\footnotesize{Alg $1$}};

\addplot[solid,line width=1.0pt]
 plot [error bars/.cd, y dir = both, y explicit, error bar style={line width=1pt},   error mark options={
      rotate=90,
      line width=1pt}]
 table[row sep=crcr, y error plus index=2, y error minus index=3]{%
50	38.8	11.7456374880208	11.7456374880208\\
60	60.3	42.2470117286418	42.2470117286418\\
70	59.4	29.8167738026769	29.8167738026769\\
80	75	47.6969600708473	47.6969600708473\\
90	69.6	32.3425416440947	32.3425416440947\\
100	50.9	19.423954283307	19.423954283307\\
110	63.7	21.8771570365073	21.8771570365073\\
120	77.4	43.1490440218552	43.1490440218552\\
130	63.5	22.477766792989	22.477766792989\\
140	94.4	50.5711380136931	50.5711380136931\\
150	74.7	20.0800896412342	20.0800896412342\\
160	77.5	24.2002066106883	24.2002066106883\\
170	90.3	25.6009765438743	25.6009765438743\\
180	71.2	18.4325798519903	18.4325798519903\\
190	93	38.5201246103903	38.5201246103903\\
200	79.6	24.654411369976	24.654411369976\\
};
\addlegendentry{\footnotesize Alg $2$};
\end{axis}
\end{tikzpicture}
\caption{\small Primal (top) and dual (bottom) updates required to converge; mean and standard deviation for $10$ repetitions. As each step of Algorithm~\ref{alg:asp} performs one primal and one dual update, the two black lines (top and bottom) coincide.}
\label{fig:iterations_x}
\end{center}
\end{figure}
\section{Route choice in a road network}
\label{sec:traffic}
As second application we study a population of drivers interacting in a road network. Our model differs from~\cite{correa2011wardrop} in the cost function~\eqref{eq:cost_traffic}, where we introduce a term penalizing the deviation from a preferred route. We assume that the travel time on each road depends only on the traffic on that road, whereas~\cite{dafermos1980traffic} considers also upstream and downstream influence. While most traffic literature focuses solely on the Wardrop equilibrium~\cite{correa2011wardrop,dafermos1980traffic}, we also study the Nash equilibrium and illustrate the distance between the two.

We consider a strongly-connected directed graph $(\mc{V},\mc{E})$ with vertex set $\mc{V} = \{1,\dots,V\}$, representing geographical locations, and directed edge set $\mc{E} = \{1,\dots,E\} \subseteq \mc{V} \times \mc{V}$, representing roads connecting the locations. Each agent $i \in \{1,\dots,\N\}$ represents a driver who wants to drive from his origin $o^i \in \mc{V}$ to his destination $d^i \in \mc{V}$.
\subsection*{Constraints}
Let us introduce the vector $x^i \in [0,1]^E$ to describe the strategy (route choice) of agent $i$, with $[x^i]_e$ representing the probability that agent $i$ transits on edge $e$~\cite{de2005route}. To guarantee that agent $i$ leaves his origin and reaches his destination with probability 1, the strategy $x\i$ has to satisfy
\begin{equation}
\sum_{e \in \text{in}(v)} [x\i]_e - \sum_{e \in \text{out}(v)} [x\i]_e =
\begin{cases}
-1 &\text{if} \quad v = o^i \\
1 &\text{if} \quad v = d^i \\
0 &\text{otherwise},
\end{cases} \qquad \forall \; v \in \mc{V},
\label{eq:traffic_constraints}
\end{equation}
where $\text{in}(v)$ and $\text{out}(v)$ represent the set of in-edges and the set of out-edges of node $v$. We denote  the graph incidence matrix by $B \in \mb{R}^{V \times E}$,  so that $[B]_{ve} = 1$ if edge $e$ points to vertex $v$, $[B]_{ve} = -1$ if edge $e$ exits vertex $v$ and $[B]_{ve} = 0$ otherwise. The individual constraint set of agent $i$ is then
\begin{equation}
\mc{X}^i \defeq \{x \in [0,1]^E : Bx = b\i \},
\label{eq:local_traffic}
\end{equation}
where $b^i \in \mb{R}^V$ is such that $[b^i]_v = -1$ if $v = o^i$,  $[b^i]_v = 1$ if $v = d^i$ and $[b^i]_v = 0$ otherwise. We introduce the constraint
\begin{equation}
x\in\mc{C}\defeq\{x\in\R^{\N E}\mid \textstyle \frac{1}{\N}\sum_{i=1}^\N x\i_e \le K_e,  \, \forall \, e=1,\dots,E\},
\label{eq:coupling_global_traffic}
\end{equation}
expressing the fact that the number of vehicles on edge $e$ cannot exceed $\N K_e$. Such coupling constraint can be imposed by authorities to decrease the congestion in a specific road or neighborhood, with the goal of reducing noise or pollution.
\subsection*{Cost function}
We assume that each driver $i\in\{1,\dots,\N\}$ wants to minimize his travel time and, at the same time, does not want to deviate too much from a preferred route $\tilde x^i \in \mc{X}\i$. We model this objective with the following cost function 
\begin{equation}
J^i(x\i,\sigma(x)) = \frac{\gamma^i}{2} \|x^i-\tilde x^i\|^2 + \sum_{e=1}^E t_e(\sigma_e(x_e)) x\i_e,
\label{eq:cost_traffic}
\end{equation}
with $\gamma^i\ge0$ a weighting factor, $x_e \defeq [x^1_e,\ldots,x^\N_e]^\top$, $\sigma_e(x_e) = \frac{1}{\N} \sum_{i=1}^{\N} x\i_e$ and $t_e(\sigma_e(x_e))$ the travel time on edge $e$.
\subsection*{Travel time}
This subsection is devoted to the derivation of the analytical expression of the travel time $t_e(\sigma_e(x_e))$. The reader not interested in the technical details of the derivation can jump to the expression of $t_e(\sigma_e(x_e))$ in~\eqref{eq:travel_time_smoothed}, which is illustrated in Figure~\ref{fig:t_cont_diff}.
We introduce the quantity $D_e(x_e) = \sum_{i=1}^\N x\i_e$ to describe the total demand on edge $e$. We consider a rush-hour interval $[0,h]$ and we assume that the instantaneous demand equals ${D_e(x_e)}/{h}$ at any time $t\in[0,h]$ and zero for $t > h$. We assume that edge $e$ can support a maximum flow $F_e$ (vehicles per unit of time) and features a free-flow travel time $t_{e,\text{free}}$. As we are interested in comparing populations of different sizes, we further assume that the peak hour duration $h$ is independent from the population size $\N$ and that  the road maximum capacity flow $F_e$ scales linearly with the population size, i.e. $F_e(\N) = f_e \cdot \N$, with $f_e$ constant in $\N$. The consideration underpinning this last assumption is that the road infrastructure scales with the number of vehicles to accommodate the increasing demand, similarly as what assumed in~\cite{ma2013decentralized} for the energy infrastructure.

If $D_e(x_e)/h\le F_e$ then every car has instantaneous access to edge $e$ and no queue accumulates, hence the travel time equals $t_{e,\text{free}}$. We focus in the rest of this paragraph on the case $D_e(x_e) / h> F_e$. An increasing queue forms in the interval $[0,h]$ and decreases at rate $F_e$ for $t>h$. The number of vehicles $q_e(t)$ queuing on edge $e$ at time $t$ obeys then the dynamics
\begin{equation}
\dot q_e(t) =  \begin{cases} \frac{D_e(x_e)}{h}\cdot \boldsymbol{1}_{[0,h]}(t) - F_e & \text{if} \; q_e(t) \ge 0 \\
\; 0  & \text{otherwise}, \end{cases} \quad q_e(0) = 0,
\label{eq:queue_dynamics}
\end{equation}
where $\boldsymbol{1}_{[0,h]}$ is the indicator function of $[0,h]$. The solution $q_e(t)$ to~\eqref{eq:queue_dynamics} is hence
\begin{equation}
q_e(t) =  \begin{cases} \left(\frac{D_e(x_e)-F_e h}{h}\right) t \quad &\text{if} \; 0 \le t \le h \\
D_e(x_e) - F_e \, t &\text{if} \; h \le  t \le D_e(x_e) / F_e \\
\; 0  & \text{if} \; t \ge D_e(x_e) / F_e. \end{cases}
\label{eq:queue_expression}
\end{equation}
As a consequence, the total queuing time at edge $e$ (i.e, the queuing times summed over all vehicles) is the integral of $q_e(t)$, which equals $D_e(x_e) (D_e(x_e)-F_e h) / (2F_e)$; the queuing time is then $(D_e(x_e)-F_e h)/(2F_e)$.

Since $\sigma_e(x_e) \! = \! \frac 1\N \sum_{i=1}^\N x\i_e \! = \! \frac 1\N D_e(x_e)$, the travel time is
\begin{equation}
t_e^\textup{PWA}(\sigma_e(x_e)) =  \begin{cases} t_{e,\text{free}} & \text{if} \; \sigma_e(x_e) \le f_e h \\
t_{e,\text{free}} + \frac{\sigma_e(x_e)-f_e h}{2f_e}  & \text{otherwise,} \end{cases}
\label{eq:travel_time}
\end{equation}
and is reported in Figure~\ref{fig:t_cont_diff}. Note that $t_e^\textup{PWA}$ is a continuous and piece-wise affine function of $\sigma_e(x_e)$, but it is not continuously differentiable, hence Assumption~\ref{A1} would not hold. Therefore, we define $t_e$ appearing in~\eqref{eq:cost_traffic} as the smoothed version of $t_e^\textup{PWA}$
\begin{equation}
t_e(\sigma_e(x_e)) \! = \! \begin{cases} t_{e,\text{free}} & \text{if} \; \sigma_e(x_e) \le f_e h - \Delta_e \\
t_{e,\text{free}} + \frac{\sigma_e(x_e)-f_e h}{2f_e} & \text{if} \; \sigma_e(x_e) \ge f_e h + \Delta_e \\
a \sigma_e(x_e)^2 \! + \! b \sigma_e(x_e) \! + \! c \!  & \text{otherwise,} \end{cases}
\label{eq:travel_time_smoothed}
\end{equation}
where the values of $\Delta_e$, $a$, $b$, $c$ are such that $t_e$ is continuously differentiable\footnote{The values are $\Delta_e = 0.5(\sqrt{(f_e h)^2+4f_e h} - f_e h)$, $a = 1/(8f_e\Delta_e)$, $b=1/(4f_e)-h/(4 \Delta_e)$, $c = t_\textup{e,free} + (f_e h)^2/(8f_e \Delta_e) - h/4 - (\Delta_e)/(8f_e)$.}, as illustrated in Figure~\ref{fig:t_cont_diff}.
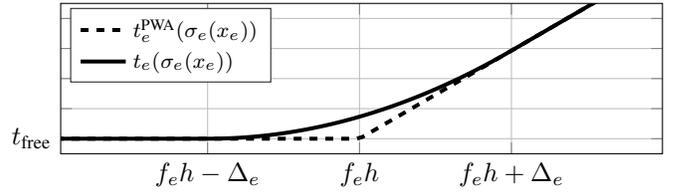
\begin{figure}[h]
\begin{center}
%
%
\begin{tikzpicture}

\begin{axis}[%
width=8cm,
height=2cm,
at={(1.011in,0.642in)},
scale only axis,
xmin=0.56,
xmax=3.46,
xtick={1.2679,2,2.7321},
xticklabels={{$f_e h-\Delta_e$},{$f_e h$},{$f_e h+\Delta_e$}},
xmajorgrids,
ymajorgrids,
xminorgrids,
yminorgrids,
ymin=1.75,
ymax=4.25,
ytick={2,2.5,3,3.5,4},
yticklabels={{$t_{\text{free}}$}},
axis background/.style={fill=white},
legend style={at={(0.35,0.97)}, font=\footnotesize, legend cell align=left, align=left, draw=white!15!black}
]

\addplot [color=black, dashed, line width=1.5pt]
  table[row sep=crcr]{%
0.535898384862246	2\\
0.565476195269069	2\\
0.595054005675892	2\\
0.624631816082716	2\\
0.654209626489539	2\\
0.683787436896362	2\\
0.713365247303186	2\\
0.742943057710009	2\\
0.772520868116832	2\\
0.802098678523655	2\\
0.831676488930479	2\\
0.861254299337302	2\\
0.890832109744125	2\\
0.920409920150949	2\\
0.949987730557772	2\\
0.979565540964595	2\\
1.00914335137142	2\\
1.03872116177824	2\\
1.06829897218507	2\\
1.09787678259189	2\\
1.12745459299871	2\\
1.15703240340554	2\\
1.18661021381236	2\\
1.21618802421918	2\\
1.24576583462601	2\\
1.27534364503283	2\\
1.30492145543965	2\\
1.33449926584648	2\\
1.3640770762533	2\\
1.39365488666012	2\\
1.42323269706695	2\\
1.45281050747377	2\\
1.48238831788059	2\\
1.51196612828742	2\\
1.54154393869424	2\\
1.57112174910106	2\\
1.60069955950789	2\\
1.63027736991471	2\\
1.65985518032153	2\\
1.68943299072836	2\\
1.71901080113518	2\\
1.748588611542	2\\
1.77816642194883	2\\
1.80774423235565	2\\
1.83732204276247	2\\
1.8668998531693	2\\
1.89647766357612	2\\
1.92605547398294	2\\
1.95563328438977	2\\
1.98521109479659	2\\
2.01478890520341	2.02957781040682\\
2.04436671561024	2.08873343122047\\
2.07394452601706	2.14788905203412\\
2.10352233642388	2.20704467284776\\
2.1331001468307	2.26620029366141\\
2.16267795723753	2.32535591447506\\
2.19225576764435	2.3845115352887\\
2.22183357805118	2.44366715610235\\
2.251411388458	2.502822776916\\
2.28098919886482	2.56197839772964\\
2.31056700927164	2.62113401854329\\
2.34014481967847	2.68028963935694\\
2.36972263008529	2.73944526017058\\
2.39930044049211	2.79860088098423\\
2.42887825089894	2.85775650179788\\
2.45845606130576	2.91691212261152\\
2.48803387171258	2.97606774342517\\
2.51761168211941	3.03522336423882\\
2.54718949252623	3.09437898505246\\
2.57676730293305	3.15353460586611\\
2.60634511333988	3.21269022667976\\
2.6359229237467	3.2718458474934\\
2.66550073415352	3.33100146830705\\
2.69507854456035	3.3901570891207\\
2.72465635496717	3.44931270993434\\
2.75423416537399	3.50846833074799\\
2.78381197578082	3.56762395156164\\
2.81338978618764	3.62677957237528\\
2.84296759659446	3.68593519318893\\
2.87254540700129	3.74509081400258\\
2.90212321740811	3.80424643481622\\
2.93170102781493	3.86340205562987\\
2.96127883822176	3.92255767644352\\
2.99085664862858	3.98171329725716\\
3.0204344590354	4.04086891807081\\
3.05001226944223	4.10002453888446\\
3.07959007984905	4.1591801596981\\
3.10916789025587	4.21833578051175\\
3.1387457006627	4.2774914013254\\
3.16832351106952	4.33664702213904\\
3.19790132147634	4.39580264295269\\
3.22747913188317	4.45495826376634\\
3.25705694228999	4.51411388457998\\
3.28663475269681	4.57326950539363\\
3.31621256310364	4.63242512620728\\
3.34579037351046	4.69158074702092\\
3.37536818391728	4.75073636783457\\
3.40494599432411	4.80989198864822\\
3.43452380473093	4.86904760946186\\
3.46410161513775	4.92820323027551\\
};
\addlegendentry{$t_e^\textup{PWA} (\sigma_e(x_e))$}

\addplot [color=black, line width=1.5pt]
  table[row sep=crcr]{%
0.535898384862246	2\\
0.565476195269069	2\\
0.595054005675892	2\\
0.624631816082716	2\\
0.654209626489539	2\\
0.683787436896362	2\\
0.713365247303186	2\\
0.742943057710009	2\\
0.772520868116832	2\\
0.802098678523655	2\\
0.831676488930479	2\\
0.861254299337302	2\\
0.890832109744125	2\\
0.920409920150949	2\\
0.949987730557772	2\\
0.979565540964595	2\\
1.00914335137142	2\\
1.03872116177824	2\\
1.06829897218507	2\\
1.09787678259189	2\\
1.12745459299871	2\\
1.15703240340554	2\\
1.18661021381236	2\\
1.21618802421918	2\\
1.24576583462601	2\\
1.27534364503283	2.00003734572021\\
1.30492145543965	2.00093364300527\\
1.33449926584648	2.00302500333706\\
1.3640770762533	2.0063114267156\\
1.39365488666012	2.01079291314087\\
1.42323269706695	2.01646946261289\\
1.45281050747377	2.02334107513165\\
1.48238831788059	2.03140775069714\\
1.51196612828742	2.04066948930938\\
1.54154393869424	2.05112629096836\\
1.57112174910106	2.06277815567408\\
1.60069955950789	2.07562508342654\\
1.63027736991471	2.08966707422574\\
1.65985518032153	2.10490412807167\\
1.68943299072836	2.12133624496435\\
1.71901080113518	2.13896342490377\\
1.748588611542	2.15778566788993\\
1.77816642194883	2.17780297392284\\
1.80774423235565	2.19901534300248\\
1.83732204276247	2.22142277512886\\
1.8668998531693	2.24502527030198\\
1.89647766357612	2.26982282852184\\
1.92605547398294	2.29581544978844\\
1.95563328438977	2.32300313410179\\
1.98521109479659	2.35138588146187\\
2.01478890520341	2.38096369186869\\
2.04436671561024	2.41173656532226\\
2.07394452601706	2.44370450182256\\
2.10352233642388	2.4768675013696\\
2.1331001468307	2.51122556396339\\
2.16267795723753	2.54677868960391\\
2.19225576764435	2.58352687829118\\
2.22183357805118	2.62147013002519\\
2.251411388458	2.66060844480593\\
2.28098919886482	2.70094182263342\\
2.31056700927164	2.74247026350764\\
2.34014481967847	2.78519376742861\\
2.36972263008529	2.82911233439632\\
2.39930044049211	2.87422596441077\\
2.42887825089894	2.92053465747195\\
2.45845606130576	2.96803841357988\\
2.48803387171258	3.01673723273455\\
2.51761168211941	3.06663111493596\\
2.54718949252623	3.11772006018411\\
2.57676730293305	3.170004068479\\
2.60634511333988	3.22348313982063\\
2.6359229237467	3.278157274209\\
2.66550073415352	3.33402647164411\\
2.69507854456035	3.39109073212596\\
2.72465635496717	3.44935005565455\\
2.75423416537399	3.50846833074799\\
2.78381197578082	3.56762395156164\\
2.81338978618764	3.62677957237528\\
2.84296759659446	3.68593519318893\\
2.87254540700129	3.74509081400258\\
2.90212321740811	3.80424643481622\\
2.93170102781493	3.86340205562987\\
2.96127883822176	3.92255767644352\\
2.99085664862858	3.98171329725716\\
3.0204344590354	4.04086891807081\\
3.05001226944223	4.10002453888446\\
3.07959007984905	4.1591801596981\\
3.10916789025587	4.21833578051175\\
3.1387457006627	4.2774914013254\\
3.16832351106952	4.33664702213904\\
3.19790132147634	4.39580264295269\\
3.22747913188317	4.45495826376634\\
3.25705694228999	4.51411388457998\\
3.28663475269681	4.57326950539363\\
3.31621256310364	4.63242512620728\\
3.34579037351046	4.69158074702092\\
3.37536818391728	4.75073636783457\\
3.40494599432411	4.80989198864822\\
3.43452380473093	4.86904760946186\\
3.46410161513775	4.92820323027551\\
};
\addlegendentry{$t_e (\sigma_e(x_e))$}

\end{axis}
\end{tikzpicture}%
\end{center}
\caption{\small Piece-wise affine travel time $t_e^\textup{PWA}(\sigma_e(x_e))$ and its smooth approximation $t_e(\sigma_e(x_e))$ as functions of $\sigma_e(x_e)$.}
\label{fig:t_cont_diff}
\end{figure}
We note that the function $t_e(\sigma_e(x_e))$ is used within a stationary traffic model but includes the average queuing time which is based on the dynamic function~\eqref{eq:queue_expression}. A thorough analysis of a dynamic traffic model is subject of future work.

Finally, we remark that a travel time with similar monotonicity properties can be derived from the piecewise affine fundamental diagram of traffic~\cite[Figure 7]{FDLi}, but $t_e(\sigma_e(x_e))$ would present a vertical asymptote which is absent here.
\subsection{Theoretical guarantees}
We  define the route-choice game $\mc{G}^\text{RC}_\N$ as in~\eqref{eq:GNEP}, with $\mathcal{X}\i$~as in \eqref{eq:local_traffic}, $\mathcal{C}$ as in~\eqref{eq:coupling_global_traffic} and $J\i(x\i,\sigma(x))$ as in~\eqref{eq:cost_traffic},~\eqref{eq:travel_time_smoothed}. In the following we apply the main results of Sections~\ref{sec:connection_VI},~\ref{sec:Wardrop_Nash},~\ref{sec:algorithms} to the route choice game.
\begin{corollary}
\label{cor:traffic1}
Consider the sequence of games $(\mc{G}^\textup{RC}_M)_{M=1}^\infty$.
Assume that for each game $\mc{G}^\textup{RC}_M$ the set $\mc{Q}=\mc{C}\cap\mc{X}$ is non-empty, that $h > 0$ and $t_\textup{e,free}, f_e > 0$ for each $e \in \mc{E}$. Moreover, assume that there exists $\hat \gamma > 0$ such that $\gamma\i \ge \hat \gamma$ for all $i \in \{1,\dots,\N\}$, for all $\N$.
Then:
\begin{enumerate}
\item The operator $F\WE$ is strongly monotone, hence each game $\mathcal{G}^\textup{RC}_M$ admits a unique variational Wardrop equilibrium. For every $\N$ satisfying
\begin{equation}
M > \maxx{e \in \mc{E}} \frac{1}{32 f_e \Delta_e \hat \gamma}
\label{eq:bound_M_traffic}
\end{equation}
the operator $F\NE$ is strongly monotone, hence each game $\mathcal{G}^\textup{RC}_M$ admits a unique variational Nash equilibrium. Every Wardrop equilibrium is an $\varepsilon$-Nash equilibrium with $\varepsilon=\frac{E}{\N f_\text{min}}$, where $f_\textup{min} = \min_{e \in \mc{E}} f_e$.
\item  For any variational Nash equilibrium $\VNE{x}$ of $\mathcal{G}^\textup{RC}_M$, the unique variational Wardrop equilibrium $\VWE{x}$ of $\mathcal{G}^\textup{RC}_M$ satisfies
\begin{align}
\| \VNE{x} - \VWE{x} \| \le \frac{\sqrt{E}}{2 f_\textup{min} \hat \gamma {\sqrt \N}}.
\label{eq:convergence_strategies_traffic}
\end{align}
\item For any $\N$, Algorithm~\ref{alg:asp} with operator $F\WE$ converges to a variational Wardrop equilibrium of $\mathcal{G}^\textup{RC}_\N$. For $\N$ satisfying~\eqref{eq:bound_M_traffic}, Algorithm~\ref{alg:asp} with operator $F\NE$ converges to a variational Nash equilibrium of $\mathcal{G}^\textup{RC}_\N$.
\qed
\end{enumerate}
\end{corollary}
\begin{proof}
1) Assumption~\ref{A1} and the consequent existence of a variational Nash and of a variational Wardrop equilibrium for any $\N$ can be shown as in Corollary~\ref{cor:pev}. The operator $F\WE$ for the cost~\eqref{eq:cost_traffic} reads
\begin{equation}
F\WE(x)=[\gamma^i(x^i-\hat x^i) + t(\sigma(x))]_{i=1}^\N.
\end{equation}
where $t(\sigma(x))\defeq[t_e(\sigma_e(x_e))]_{e=1}^E$. Since $t_e(\sigma_e(x_e))$ in~\eqref{eq:travel_time_smoothed} is a monotone function of $\sigma_e(x_e)$, the operator $t(\sigma(x))$ is monotone. Then $F\WE$ is strongly monotone with constant $\hat \gamma$ because it is the sum of a monotone and a strongly monotone operator with constant $\hat \gamma$. As a consequence, each $\mathcal{G}^\textup{RC}_\N$ admits a unique variational Wardrop equilibrium.

To prove strong monotonicity of $F\NE$ we use the result of Lemma~\ref{lemma:pd}\footnote{Lemma~\ref{lemma:pd} requires $F\NE$ to be continuously differentiable, which is not the case here. The more general result~\cite[Proposition 2.1]{schaible1996generalized} extends the statement of Lemma~\ref{lemma:pd} to operators which are not continuously differentiable. It then suffices to show $\nabla_x F\NE(x) \succ 0$ for $\sigma(x)$ in each of the three intervals defined by~\eqref{eq:travel_time_smoothed}, because in each of them $F\NE$ is continuously differentiable.}. We first note that each $t_e$ only depends on the corresponding $\sigma_e$, hence $\nabla_x F\NE(x)$ can be permuted into diagonal form similarly to what done in~\eqref{eq:PEV_gradient_FN}. It then suffices to show $\hat \gamma I_\N + \frac{1}{\N} t_e'(\sigma_e) I_\N + \frac{1}{\N^2} t''_e(\sigma_e) x_e \ones[\N]^\top \succ 0$ for all $\sigma_e$ and for all $e$. This matrix is indeed positive definite if $\sigma_e(x_e) \notin [f_e h - \Delta_e, f_e h + \Delta_e]$, because then $t'_e(\sigma_e) \ge 0$ and $t''_e(\sigma_e) = 0$ by~\eqref{eq:travel_time_smoothed}. For $\sigma_e(x_e) \in [f_e h - \Delta_e, f_e h + \Delta_e]$ it suffices to show $\hat \gamma I_\N + \frac{1}{\N^2 4f_e \Delta_e} x_e \ones[\N]^\top \succ 0$, because $t_e'(\sigma_e) \ge 0$ and $t''_e(\sigma_e) = \frac{1}{4f_e \Delta_e}$. By Lemma~\ref{lem:min_eigenval} in the Appendix, $\lambda_\text{min} \left( x_e \ones[\N]^\top + \ones[\N] x_e^\top \right)/2 \ge -\frac{\N}{8}$, which proves strong monotonicity of $F\NE$ under~\eqref{eq:bound_M_traffic}.
Consequently, if $M$ satisfies~\eqref{eq:bound_M_traffic} then $\mathcal{G}^\textup{RC}_\N$ admits a unique variational Nash equilibrium.
Finally, we verify Assumption~\ref{A2} in order to use Proposition~\ref{prop:conv_cost}. We have $\mc{X}^0 = [0,1]^E$ and $t$ is continuously differentiable and hence Lipschitz in $\mc{X}^0$, with constant $L_p = 1/(2f_\textup{min})$. Moreover, $\! \Dx \! \defeq \! \max_{y \in{\mathcal{X}\zero} } \{\|y \|\} = \sqrt{E}$. Using~\eqref{eq:Lipschitz_implies_Lipschitz} concludes the proof.

2) Since all the assumptions of Theorem~\ref{thm:conv_strategies} have just been verified, it is a direct consequence of its second statement.

3) As Assumption~\ref{A3} holds trivially (the others have already been verified), we apply Theorem~\ref{thm:convergence_asp} and conclude the proof.
\end{proof}
\subsection{Numerical analysis}
For the numerical analysis we use the data set of the city of Oldenburg~\cite{OldenburgDataset}, whose graph features 175 nodes and 213 undirected edges\footnote{The graph in the original data set features 6105 vertexes and 7035 undirected edges. We reduce it by excluding all the nodes that are outside the rectangle $[3619,4081] \times [3542,4158]$ and all the edges that do not connect two nodes in the rectangle. The resulting graph is strongly connected.} and is reported in Figure~\ref{fig:queuing_time_portion}. For each agent $i$ the origin $o\i$ and the destination $d\i$ are chosen uniformly at random. Regarding the cost~\eqref{eq:cost_traffic}, $t_{e,\text{free}}$ is computed as the ratio between the road length, which is provided in the data set, and the free-flow speed.
Based on the road topology, we divide the roads into main roads, where the free-flow speed is $50$ km/h, and secondary roads, where the free-flow speed is $30$ km/h. Moreover, we assume a peak hour duration $h$ of $2$ hours, and for all $e \in \mc{E}$, we set $f_e = 4\cdot10^{-3}$ vehicles per second, which corresponds to 1 vehicle every 4 seconds for a population of $\N = 60$ vehicles. Finally, the parameter $\gamma\i$ is picked uniformly at random in $[0.5,3.5]$ and $\tilde x\i$ is such that $\tilde x\i_e = 1$ if $e$ belongs to the shortest path from $o\i$ to $d\i$, while $\tilde x\i_e = 0$ otherwise. The shortest path is computed based on $\{t_{e,\text{free}}\}_{e=1}^E$. Note that with the above values the bound~\eqref{eq:bound_M_traffic} becomes $M > 16.14$, which is satisfied also for small-size populations.

We compute the Wardrop equilibrium with Algorithm~\ref{alg:asp} relatively to a population of $\N = 60$ drivers without coupling constraint, i.e. with $K_e = 1$ for all $e \in \mc{E}$. We report in Figure~\ref{fig:queuing_time_portion} the corresponding queuing time $t_e(\sigma_e(x_e)) - t_{e,\text{free}}$ as by~\eqref{eq:travel_time_smoothed}.
\begin{figure}[h]
\begin{center}
\includegraphics[width = 0.75\columnwidth]{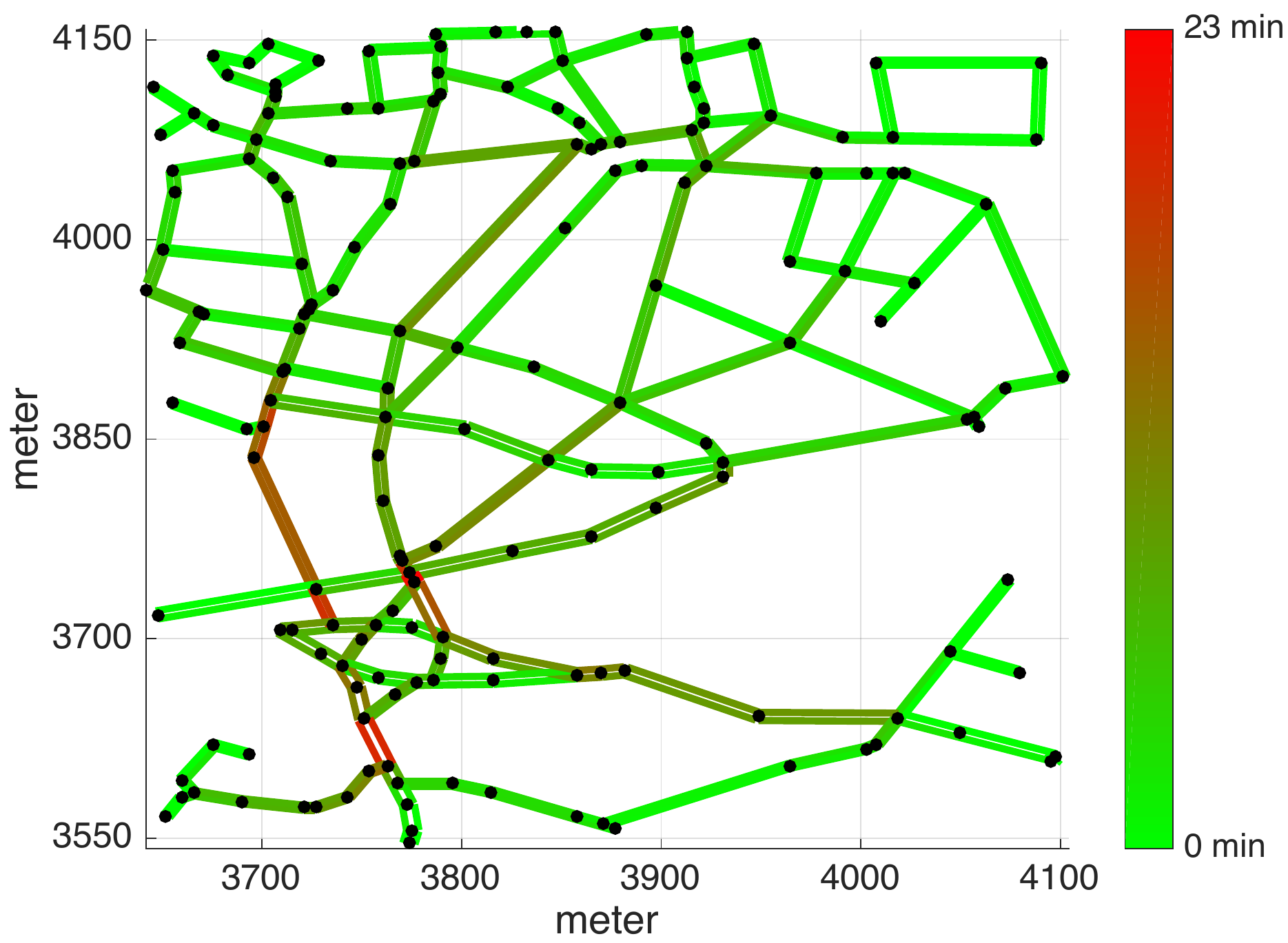}
\end{center}
\vspace{-0.1cm}
\caption{\small The queuing time reported in green-red color scale.
Note that this pattern changes if one modifies the pairs origin-destination.}
\label{fig:queuing_time_portion}
\end{figure}
We illustrate in Figure~\ref{fig:comparison_coupling} the change in the queuing time of an entire neighborhood when introducing a coupling constraint that upper bounds the total number of cars on a single edge, relatively to a Wardrop equilibrium with $\N = 60$.
\begin{figure}[h]
\begin{center}
\includegraphics[width = 1\columnwidth]{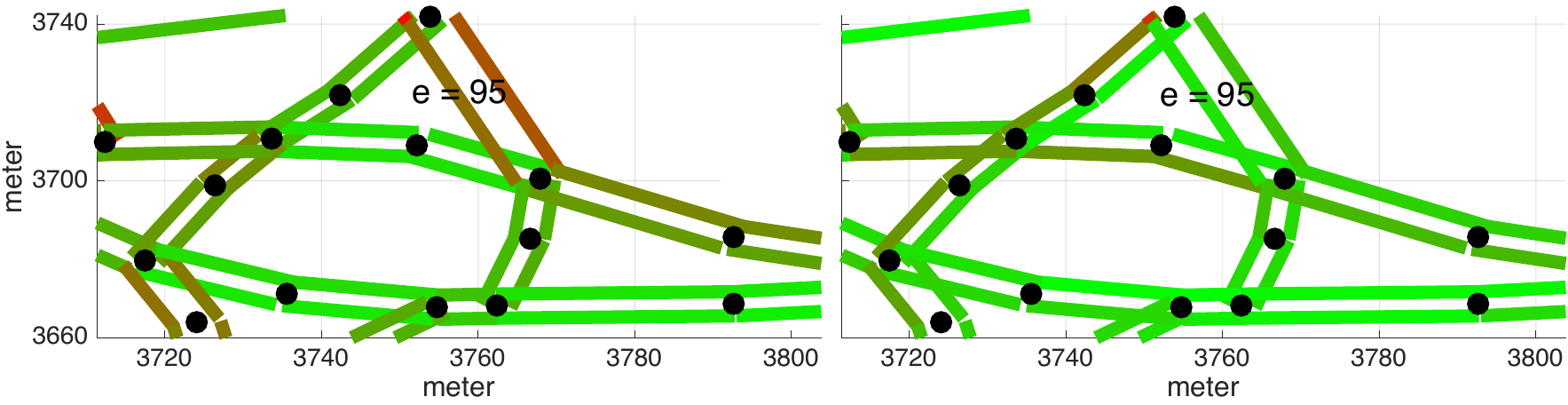}
\end{center}
\vspace{-0.1cm}
\caption{\small On the left, the queuing time in a neighborhood without any coupling constraints; 10\% of the population transits on edge 95, and the queuing time is 7.28 minutes. On the right, the queuing time in presence of a coupling constraint allowing at most 3\% of the entire population on edge 95; the queuing time is reduced to 1.42 minutes, but it visibly increases on the edges of the alternative route.}
\label{fig:comparison_coupling}
\end{figure}

Finally, we illustrate the second statement of Corollary~\ref{cor:traffic1} by reporting in Figure~\ref{fig:traffic_conv_x} the distance between the unique variational Wardrop equilibrium and the variational Nash equilibrium found by Algorithm~\ref{alg:asp}. The $\varepsilon$-Nash property of the Wardrop equilibrium in Proposition~\ref{prop:conv_cost} can also be illustrated, but a plot is omitted here for reasons of space.

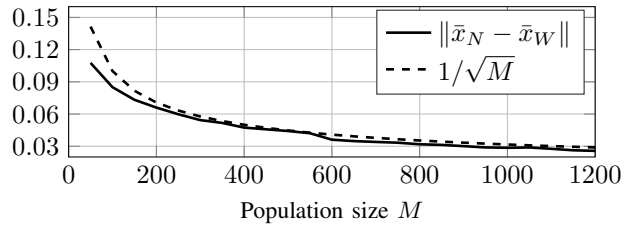
\begin{figure} [h!]
\begin{center}
\setlength\figureheight{2.0cm} 
\setlength\figurewidth{7.cm} 
%
%
\begin{tikzpicture}

\begin{axis}[%
width=\figurewidth,
height=\figureheight,
at={(1.011in,0.642in)},
scale only axis,
xmin=0,
xmax=1200,
tick label style={/pgf/number format/.cd,1000 sep={}},
xtick={0, 200, 400, 600, 800, 1000, 1200},
ymin=0.02,
ymax=0.16,
tick label style={/pgf/number format/fixed},
ytick={0,0.03,0.06,0.09,0.12,0.15},
axis background/.style={fill=white},
xmajorgrids,
ymajorgrids,
xlabel ={\small Population size $M$},
legend style={legend cell align=left, align=left, draw=white!15!black}
]
\addplot [color=black, line width=1.0pt]
  table[row sep=crcr]{%
50	0.107597186907606\\
100	0.0850078973215346\\
150	0.0732622497727821\\
200	0.0660039393423034\\
250	0.0597843249632177\\
300	0.0543514624193683\\
350	0.0517201549776713\\
400	0.047375072456806\\
450	0.0457663821274737\\
500	0.0441920556441806\\
550	0.0421236118876949\\
600	0.0361\\
650	0.0348\\
700	0.034\\
750	0.0333\\
800	0.0318333333333333\\
850	0.0312666666666667\\
900	0.0301\\
950	0.0289\\
1000	0.0285666666666667\\
1050	0.0288666666666667\\
1100	0.0276333333333333\\
1150	0.0262333333333333\\
1200	0.0257\\
};
\addlegendentry{$\| \bar x_N - \bar x_W \|$}

\addplot [color=black, dashed, line width=1.0pt]
  table[row sep=crcr]{%
50	0.14142135623731\\
100	0.1\\
150	0.0816496580927726\\
200	0.0707106781186548\\
250	0.0632455532033676\\
300	0.0577350269189626\\
350	0.0534522483824849\\
400	0.05\\
450	0.0471404520791032\\
500	0.0447213595499958\\
550	0.0426401432711221\\
600	0.0408248290463863\\
650	0.0392232270276368\\
700	0.0377964473009227\\
750	0.0365148371670111\\
800	0.0353553390593274\\
850	0.0342997170285018\\
900	0.0333333333333333\\
950	0.0324442842261525\\
1000	0.0316227766016838\\
1050	0.0308606699924184\\
1100	0.0301511344577764\\
1150	0.0294883912309794\\
1200	0.0288675134594813\\
};
\addlegendentry{$1/\sqrt{M}$}

\end{axis}
\end{tikzpicture}%
\caption{\small Distance between Nash and Wardrop variational equilibria. As in Fig. 2, $1/\sqrt{M}$ illustrates the trend of the bound derived in Corollary~\ref{cor:traffic1} and not the specific constant.}
\label{fig:traffic_conv_x}
\end{center}
\end{figure}

\section{Conclusions}
\label{sec:conclusions}

The paper considered aggregative games and established novel results on the Euclidean distance between Nash and Wardrop equilibrium; moreover, it proposed two decentralized algorithms to achieve the two equilibria in presence of coupling constraints and investigated two relevant applications.
As future research direction, it would be interesting to design distributed algorithms which achieve an equilibrium by means of local communications. Moreover, by exploiting the VI reformulation one could establish results on the proximity between Nash equilibrium and social optimum.
\section*{Appendix: Proofs}
\subsection*{Proof of Lemma \ref{lem:FNstrongly monotone}}
1) Let us first show that $F\WE$ is monotone.
Since $v\i$ is convex, then $\nabla_{x\i} v\i(x\i)$ is monotone in $x\i$ by~\cite[Section 4.2.2]{scutari2012monotone}.
Hence $[\nabla_{x\i} v\i(x\i)]_{i=1}^\N$ is monotone.
Moreover, for any $x_1,x_2$
\begin{equation}
\begin{aligned}
& ([p(\sigma(x_1))]_{i=1}^\N - [p(\sigma(x_2))]_{i=1}^M)^\top (x_1 - x_2) \\
&= \N ( p(\sigma(x_1)) - p(\sigma(x_2)) )^\top ( \sigma(x_1) - \sigma(x_2) ) \ge 0,
\label{eq:p_mon}
\end{aligned}
\end{equation}
where the last inequality follows from the fact that $p$ is monotone.
By~\eqref{eq:F_W_decomp} and the fact that the sum of two monotone operators is monotone, one can conclude that $F\WE$ is monotone. \\
To show that $F\NE$ is strongly monotone, we write the affine expression of $p$ as $p(x) = C x + c$, where there exists $\alpha > 0$ such that $C \succ \alpha I_n$ by Lemma~\ref{lemma:pd}.
Then the term $\frac1 \N [  \nabla_z p(z)_{|{z=\sigma(x)}} {x\i}]_{i=1}^\N$ in~\eqref{eq:F_N_decomp} equals
$\frac{1}{\N} (I_\N \otimes C^\top) x$.
Since $\nabla_x (\frac{1}{\N} (I_\N \otimes C^\top) x) \succ \frac{\alpha}{\N} I_{\N n}$,
then $\frac1 \N [  \nabla_z p(z)_{|{z=\sigma(x)}} {x\i}]_{i=1}^\N$ is strongly monotone by Lemma~\ref{lemma:pd}.
Having already shown that $F\WE$ is monotone, the proof is concluded upon noting that the sum of a monotone operator and a strongly monotone operator is strongly monotone.

2) Strong convexity of $v^i$ is equivalent to strong monotonicity of $\nabla_{x\i}v\i(x\i)$ in $x\i$~\cite[Section 4.2.2]{scutari2012monotone}.
Then $[\nabla_{x\i} v\i(x\i)]_{i=1}^\N$ is strongly monotone.
Monotonicity of $[p(\sigma(x))]_{i=1}^\N$ in~\eqref{eq:F_W_decomp} can be shown as in~\eqref{eq:p_mon}.
\hfill $\smallblacksquare$
\subsection*{Proof of Proposition \ref{prop:ext_vi}}

Under Assumptions~\ref{A1} and~\ref{A3} the set $\mc{Q}$, and consequently the sets $\{\mc{X}^i\}_{i=1}^\N$, $\mc{X}$ and $\mc{Y}$, are convex and satisfy Slater's constraint qualification.
The VI$(\mathcal{Q},F)$ is therefore equivalent to its KKT system~\cite[Proposition 1.3.4]{facchinei2007finite}.
Moreover, since $\mc{X}^i$ satisfies Slater's constraint qualification, the optimization problem of agent $i$ in the game~\eqref{eq:GNEP_ext} is equivalent to its KKT system, for each $i$.
Finally, by~\cite[Proposition 1.3.4]{facchinei2007finite}, the VI$(\mathcal{Y},T)$ is equivalent to its KKT system.
We do not report the three KKT systems here,
but it can be seen by direct inspection that they are equivalent~\cite[Section 4.3.2]{scutari2012monotone}.
\hfill \smallblacksquare
\subsection*{Proof of Theorem \ref{thm:conv_two}}
We split the proof of the theorem into two parts. First we show convergence of the inner loop and then of the outer loop.

\textit{Inner loop}.
Using the same approach of \cite[Theorem 3 and Corollary 1]{grammatico:parise:colombino:lygeros:14}, it is possible to show that under Assumption~\ref{A4} for any $\lambda_{(k)}\in\R^m_{\ge0}$ the sequences of $z_{(h)}$ and of $\tilde x(h)$ converge respectively to $\bar z$ and to $\bar{{x}}$ such that $\bar z=\frac{1}{\N}\sum_{i=1}^\N x^i_{\textup{or}} (\bar z,\lambda_{(k)})\eqdef\frac{1}{\N}\sum_{i=1}^\N \bar{{x}}^i=\sigma(\bar{{x}})$. In~\cite[Theorem 1]{grammatico:parise:colombino:lygeros:14} it is shown that the set $\{\bar{{x}}^i\}_{i=1}^\N$ is an $\varepsilon$-Nash equilibrium for the game $\mc{G}(\lambda_{(k)})$, with $\varepsilon=\mc{O}(\frac{1}{\N})$. In the following, we show that $\{\bar{{x}}^i\}_{i=1}^\N$ is actually a Wardrop equilibrium of $\mc{G}(\lambda_{(k)})$\footnote{This is consistent with \cite[Theorem 1]{grammatico:parise:colombino:lygeros:14} thanks to Proposition~\ref{prop:conv_cost}.}. Indeed, for each agent $i$, by the definition of optimal response in~\eqref{eq:or}, one has
\[J^i(\bar x^i,\bar z) +\lambda_{(k)}^\top A_{(:,i)} \bar  x^i \le J^i(x^i,\bar z) +\lambda_{(k)}^\top A_{(:,i)}  x^i, \forall x^i\in\mc{X}^i\,.
\]
Using the fact that $\bar z=\sigma(\bar x)$, we get
\[
J^i(\bar x^i,\sigma(\bar x)) +\lambda_{(k)}^\top A_{(:,i)} \bar x^i \le J^i(x^i,\sigma(\bar x)) +\lambda_{(k)}^\top A_{(:,i)}  x^i, 
\]
for all $x^i\in\mc{X}^i$ and for all $i\in\{1,\dots,\N\}$. Thus $\{\bar x^i\}_{i=1}^\N$ is a Wardrop equilibrium of $\mc{G}(\lambda_{(k)})$ by Definition~\ref{def:WE}.

\textit{Outer loop.}
We follow the steps of the proof of~\cite[Proposition 8]{pang2010design}. For each $\lambda\in\R^m_{\ge0}$ define $F\WE(x;\lambda) \defeq F\WE(x)+A^\top\lambda$. Such operator is strongly monotone in $x$ on $\mc{Q}$ with the same constant $\alpha$ as $F\WE(x)$. It follows by Lemma \ref{lem:exun},  that $\mc{G(\lambda)}$ has a unique variational Wardrop equilibrium which we denote by $\bar x\WE(\lambda)$. Note that the outer loop update can be written as 
\[\lambda_{(k+1)}=\Pi_{\R^m_{\ge0}}[\lambda_{(k)}-\tau(b-A\bar x\WE(\lambda_{(k)}))],\]
which is a step of the projection algorithm~\cite[Algorithm 12.1.4]{facchinei2007finite} applied to VI$(\R^m_{\ge0}, \Phi)$, with $\Phi(\lambda) \defeq b-A\bar x\WE(\lambda)$.  To conclude, it suffices to show that $\lambda_{(k)}$ converges to a solution $\bar \lambda$ of such VI, because by~\cite[Proposition 1.1.3]{facchinei2007finite}, $\bar \lambda$ solves VI$(\R^m_{\ge0}, \Phi)$ if and only if $0\le \bar \lambda \perp (b-A \bar x\WE(\bar \lambda)) \ge 0$. Having already proved convergence of the inner loop, the conclusion then follows from the second statement of Proposition \ref{prop:ext_vi}.

To show that the sequence $\lambda_{(k)}$ converges to a solution of the $\textup{VI}(\R^m_{\ge0},\Phi)$, we prove that the mapping $\Phi$ is co-coercive\footnote{\label{foot:cocercive} The operator $\Phi:\R^m\rightarrow \R^m$ is co-coercive with constant $\eta>0$ if $(\Phi(\lambda_1)-\Phi(\lambda_2))^\top(\lambda_1-\lambda_2)\ge\eta ||\Phi(\lambda_1)-\Phi(\lambda_2)||^2$, for all $\lambda_1, \lambda_2 \in\R^m$.} with co-coercitivity constant $c_\Phi = \alpha/\|A\|^2$ and apply \cite[Theorem 12.1.8]{facchinei2007finite} to conclude the proof. Note that~\cite[Theorem 12.1.8]{facchinei2007finite} requires $\textup{VI}(\R^m_{\ge0},\Phi)$ to have at least a solution; this is guaranteed by the equivalence between 1) and 2) in Proposition~\ref{prop:ext_vi} upon noting that a solution of VI($Q,F$) exists by Lemma~\ref{lem:exun}.

To show co-coercitivity of $\Phi$, consider $\lambda_1 , \lambda_2 \in \R^m_{\ge0}$ and the corresponding unique solutions $x_1 \defeq \bar x\WE(\lambda_1)$ of VI($\mc{X}$,$F\WE+A^\top\lambda_1)$ and $x_2 \defeq \bar x\WE(\lambda_2)$ of VI($\mc{X}$,$F\WE+A^\top\lambda_2)$. By definition
\begin{subequations}
\begin{align}
&(x_2-x_1)^\top(F\WE(x_1)+A^\top\lambda_1) \ge 0 \label{eq:VI_1}\,, \\
&(x_1-x_2)^\top(F\WE(x_2)+A^\top\lambda_2) \ge 0 \label{eq:VI_2} \,.
\end{align}
\end{subequations}
Adding~\eqref{eq:VI_1} and~\eqref{eq:VI_2} we obtain $(x_2-x_1)^\top(F\WE(x_1)-F\WE(x_2) + A^\top (\lambda_1-\lambda_2)) \ge 0 $, i.e., $(x_2-x_1)^\top A^\top(\lambda_1-\lambda_2) \ge (x_2-x_1)^\top(F\WE(x_2)-F\WE(x_1))$.
Since $F\WE$ is strongly monotone, it follows from the last inequality that
\begin{equation}
(Ax_2-Ax_1)^\top (\lambda_1-\lambda_2) \ge \alpha \| x_2 - x_1 \|^2 \,.
\label{eq:proof_2_intermed_2}
\end{equation}
Since by definition $\|A(x_2-x_1)\| \le \| A \| \|x_2-x_1 \|$, then
\begin{equation}
\| x_2 - x_1 \|^2 \ge \frac{\| A(x_2-x_1)\|^2}{\| A \|^2} \,.
\label{eq:norm_def_ineq}
\end{equation}
Combining~\eqref{eq:proof_2_intermed_2},~\eqref{eq:norm_def_ineq}, and adding and subtracting $b$, we obtain
\begin{equation}
(b-Ax_2 - (b-Ax_1))^\top \! (\lambda_2 - \lambda_1) \ge \frac{\alpha}{\| A \|^2} \| b-Ax_2 - (b-Ax_1) \|^2 \! ,
\end{equation}
hence $\Phi$ is co-coercive in $\lambda$ with constant $c_\Phi =\alpha/\|A\|^2$. \hfill \smallblacksquare
\subsection*{Proof of Theorem \ref{thm:convergence_asp}}
We give the proof for a strongly monotone operator $F$, which is to be interpreted as $F\NE$ in the first statement and $F\WE$ in the second statement. We divide the proof into two parts: (i) we prove that Algorithm~\ref{alg:asp} is a particular case of a class of algorithms known as asymmetric projection algorithms (APA) \cite[Algorithm 12.5.1]{facchinei2007finite} applied to VI$(\mathcal{Y},T)$; (ii) we prove that our algorithm satisfies a convergence condition for APA. It can be shown that if $\tau$ satisfies~\eqref{eq:condition_tau_APA} then also $\tau<1/\|A\|$ holds. 
\\
(i) The APA are parametrized by the choice of a matrix $D\succ0$. For a fixed $D$ a step of the APA  for VI$(\mathcal{Y},T)$ is
\begin{equation}
\label{eq:APA}
y_{(k+1)}=\textup{solution of VI}(\mathcal{Y}, T^k_D),
\end{equation}
where $y_{(k)}$ is the state at iteration $k$ and $T^k_{D}(y)\defeq T(y_{(k)})+D(y-y_{(k)})$. Every step of the APA requires the solution of a different variational inequality that depends on the operator $T$, on a fixed matrix $D$ and on the previous strategies' vector $y_{(k)}$. We choose
\begin{equation}
D\coloneqq \left[\begin{array}{cc}\frac{1}{\tau} I_{\N n} & 0 \\-2A & \frac{1}{\tau} I_m\end{array}\right],
\label{eq:choice_D_apa}
\end{equation}
which by using the Schur complement condition can be shown to positive definite because $\tau<1/\|A\|$. It is shown in~\cite[Section 12.5.1]{facchinei2007finite} that with the choice~\eqref{eq:choice_D_apa} the update~\eqref{eq:APA} coincides with the steps~\eqref{eq:apa_inner}.
\\
\noindent (ii) As illustrated in the previous point, Algorithm~\ref{alg:asp} is the specific APA associated with the choice of $D$ given in \eqref{eq:choice_D_apa}. According to \cite[Proposition 12.5.2]{facchinei2007finite}, this algorithm converges if the mapping $G(y)=D_s^{-1/2}  T(D_s^{-1/2} y) -D_s^{-1/2}  (D-D_s) D_s^{-1/2} y$ is co-coercive with constant $1$, where $D_s = (D+D^\top)/2$ and $D_s^{-1/2}$ denotes the principal square root of the symmetric positive definite matrix $D_s^{-1}$ and is therefore symmetric positive definite. Let us rename $L \defeq D_s^{-1/2}$ and $Ly=\begin{sma} v\\w\end{sma}$ and simplify the expression of $G(y)$
\begin{align}
G(y)&=L  T(L y) -L  (D-D_s) L y\\
&= L \left( \begin{sma} F(v) \\ 0 \end{sma} +\begin{sma} 0 & A^\top \\ -A & 0 \end{sma}Ly +\begin{sma} 0 \\ b\end{sma} \right)-L  \begin{sma} 0 & A^\top \\ -A & 0 \end{sma} L y\\[-0.3cm]
&= L \left( \begin{sma} F(v) \\ 0 \end{sma} +\begin{sma} 0 \\ b\end{sma} \right).
\label{eq:express_G}
\end{align}
We now prove that $G(y)$ is co-coercive with constant $1$, i.e.
\begin{equation}
(y_1-y_2)^\top(G(y_1)-G(y_2)) -  \|G(y_1)-G(y_2)\|^2\ge0.
\label{eq:cocerc_repeated}
\end{equation}
Let us substitute~\eqref{eq:express_G} in the left-hand side of~\eqref{eq:cocerc_repeated}
\begin{small}
\begin{align}
&(y_1 \! -y_2)^\top \! (G(y_1)-G(y_2)) -  \|G(y_1)-G(y_2)\|^2 \\
&=(y_1 \! -y_2)^\top \! (L  \begin{sma} F(v_1) \\ 0 \end{sma} \! -L  \begin{sma} F(v_2) \\ 0 \end{sma} )  \! - \!  \|L  \begin{sma} F(v_1) \\ 0 \end{sma}  \! -L  \begin{sma} F(v_2) \\ 0 \end{sma} \|^2 \\
&=(Ly_1-Ly_2)^\top(  \begin{sma} F(v_1) -F(v_2)\\ 0 \end{sma}  ) -  \| L  \begin{sma} F(v_1) -F(v_2)\\ 0 \end{sma}  \|^2 \\
&=(\begin{sma} v_1-v_2 \\ w_1-w_2\end{sma})^\top \! (  \begin{sma} F(v_1) -F(v_2)\\ 0 \end{sma}  ) \! - \! \begin{sma} F(v_1) -F(v_2)\\ 0 \end{sma} ^\top \!\!\! L^2 \begin{sma} F(v_1) -F(v_2)\\ 0 \end{sma} \! \\
&=  ( F(v_1)  - F(v_2) )^\top [(v_1  -  v_2 )- [L^2]_{11} (F(v_1)  -  F(v_2))] \\
& \ge \alpha \|v_1-v_2\|^2 -  \|[L^2]_{11}\|\| F(v_1) -F(v_2)) \|^2 \\
& \ge \left(\alpha -  \|[L^2]_{11}\|L_F^2 \right)  \|v_1-v_2\|^2 \eqdef K \|v_1-v_2\|^2, 
\end{align}
\end{small}
The proof is concluded if $K\ge0.$ Let us compute $[L^2]_{11} = [D_s^{-1}]_{11}$. By inverting the block matrix $D_s$ we get
\begin{equation}\label{eq:l2}
[L^2]_{11} =\tau(I-\tau^2 A^\top A)^{-1} \succ 0. 
\end{equation}
Since $\tau^2 A^\top A$ is symmetric positive semidefinite, $\lambda_\text{max}(\tau^2 A^\top A) = \tau^2 \|A\|^2 <1$ because $\tau<1/\|A\|$ and  $\rho(\tau^2 A^\top A)<1$, i.e. the matrix is convergent. Hence, the Neumann series $\sum_{k=0}^\infty (\tau^2 A^\top A)^k$ converges to $(I-\tau^2 A^\top A)^{-1}$. Substituting in \eqref{eq:l2} yields $[L^2]_{11} = \tau \sum_{k=0}^\infty (\tau^2 A^\top A)^k\succeq 0$ and
$\|[L^2]_{11}\|\le \tau \sum_{k=0}^\infty (\tau^2\|A\|^2)^k= \frac{\tau}{1-\tau^2\|A\|^2},
$
where we used the fact that the geometric series converges since $\tau^2 \|A\|^2 <1$. Therefore $K\ge \alpha -  \frac{\tau}{1-\tau^2\|A\|^2} L_F^2$. By condition \eqref{eq:condition_tau_APA} we get 
$\alpha \tau^2 \|A\|^2+\tau L_F^2 <\alpha$ and thus
$
K\ge \frac{\alpha - \alpha \tau^2 \|A\|^2 - \tau L_F^2}{1-\tau^2 \|A\|^2 } >0.
$
\hfill\smallblacksquare

\begin{lemma}
\label{lem:min_eigenval}
For all $\N\in\mb{N}$, it holds 
\begin{equation}
\minn{y \in [0,1]^\N} \lambda_\textup{min} \left( y \ones[\N]^\top + \ones[\N] y^\top \right) \ge -\frac{\N}{4}.
\label{eq:min_eigenval_statement}
\end{equation}

\end{lemma}
\begin{proof}
The statement is trivially true for $\N = 1$.
For $\N > 1$, problem~\eqref{eq:min_eigenval_statement} is equivalent to
\begin{equation}
\minn{\substack{y \in [0,1]^\N \\ \| v \| = 1}}{v^\top \!\! \left( y \ones[\N]^\top + \ones[\N] y^\top \right) v} = \!\!\!\!
\minn{\substack{y \in [0,1]^\N \\ \| v \| = 1}}{2\left( v^\top y \right) \left( \ones[\N]^\top v \right)}.
\label{eq:min_eigenval_expanded}
\end{equation}
We show that~\eqref{eq:min_eigenval_expanded} is negative
by denoting $\hat y = e_1, \; \hat v = 0.6 e_1 - 0.8 e_2$ and observing that
$2 \left( \hat v^\top \hat y \right) \left( \ones[\N]^\top \hat v \right) = - 0.24$.
Let us consider a pair $y^\star, v^\star$ minimizing~\eqref{eq:min_eigenval_expanded} and note that
$\ones[\N]^\top v^\star \neq 0$, because~\eqref{eq:min_eigenval_expanded} is negative.
We are left with two cases, $\ones[\N]^\top v^\star > 0$ and $\ones[\N]^\top v^\star < 0$.
Let us start analyzing $\ones[\N]^\top v^\star > 0$.
To minimize $2\left( v^\top y \right) \left( \ones[\N]^\top v \right)$, it must be
\begin{equation}
y^\star_i = \begin{cases} 0 \quad &\text{if} \; v^\star_i > 0  \\ 
1 \quad & \text{if} \; v^\star_i < 0,
\end{cases} \;\; \text{for all} \; i \in \{1,\dots,\N\}.
\label{eq:min_eigenval_y_01}
\end{equation}
Without loss of generality, we can assume $y_i^\star \in \{0,1\}$ if $v^\star_i = 0$.
Hence we conclude that $y^\star \in \{0,1\}^\N$ and~\eqref{eq:min_eigenval_statement} reduces to
\begin{equation}
\minn{p \in \{0,\dots,\N\}}{\lambda_\textup{min} \left[
\begin{array}{c|c}
2 (\ones[p] \ones[p]^\top) & \ones[p] \ones[(\N-p)]^\top \\[0.1cm]
\hline  \\[-0.3cm]
\ones[(\N-p)] \ones[p]^\top & \zeros[(\N-p)] \zeros[(\N-p)]^\top
\end{array}
\right]},
\label{eq:min_eigenval_matrix_01}
\end{equation}
where without loss of generality we assumed the first $p$ components of $y^\star$ to be $1$ and the remaining to be $0$.
Note that the matrix in~\eqref{eq:min_eigenval_matrix_01} features $p$ identical rows followed by $\N-p$ other identical rows.
Hence any of its eigenvectors must have $p$ identical components followed by $\N-p$ other identical components.
With this observation and the definition of eigenvalue, it is easy to show that
the matrix in~\eqref{eq:min_eigenval_matrix_01} has only two distinct eigenvalues, the minimum of the two being $p-\sqrt{\N p}$.
The function $p-\sqrt{\N p}$ is minimized over the reals for $p=\N/4$ with corresponding minimum $\lambda_\textup{min} = -\N/4$,
as it can be seen by using the change of variables $p = q^2$ and minimizing the quadratic function $q^2 - \sqrt{\N} q$.
Since $p \in \{0,\dots,\N\}$ in~\eqref{eq:min_eigenval_matrix_01},
the value $-\N/4$ is a lower bound for the minimum eigenvalue, and it is attained only if $\N$ is a multiple of $4$.
We conclude by noting that the derivation for the case $\ones[\N]^\top v^\star < 0$ is identical to the derivation for the case $\ones[\N]^\top v^\star > 0$ just shown, upon switching $0$ and $1$ in~\eqref{eq:min_eigenval_y_01}.
\end{proof}
\subsection*{Proof of Proposition \ref{prop:uniqueness_for_PEVs}}
\noindent The constraints in \eqref{eq:vehicle_constraint}, \eqref{eq:coupling_global_PEVs} can be expressed as $\Gamma x \le \gamma$ with
\begin{equation}
\Gamma = \begin{sma}
 I_{\N\cdot n} \\
 -I_{\N\cdot n}  \\
  -I_\N \otimes \ones[n]^\top \\
  \ones[\N]^\top \otimes I_n
\end{sma}
,\quad \gamma=\begin{sma}
\tilde x  \\
 0\\
 -\theta \\
 \N K
\end{sma}
\,,
\end{equation}
where $\theta = [\theta^1,\dots,\theta^\N]^\top$, and $\tilde x = [[\tilde x^i_t]_{t=1}^n]_{i=1}^\N$.
Let us partition the constraint matrix $\Gamma$ into its individual part $\Gamma_1$ and coupling part $\Gamma_2$
\begin{equation}\label{eq:H_PEVs_partitioned}
\Gamma = \begin{sma} \Gamma_1 \\ \Gamma_2 \end{sma},\;
 \Gamma_1 = \begin{sma}
 I_{\N\cdot n} \\
 -I_{\N\cdot n}  \\
  -I_\N \otimes \ones[n]^\top \\
\end{sma},\;
\Gamma_2 = \begin{sma}
  \ones[\N]^\top \otimes I_n
\end{sma}
\end{equation}
and $\gamma = [\gamma_1^\top, \gamma_2^\top]^\top$ accordingly.
The KKT conditions for VI$(\mc{Q},F\NE)$ at the primal solution $\VNE{x}$ are~\cite[Proposition 1.3.4]{facchinei2007finite}
\begin{subequations}
\label{eq:KKT_PEVs}
\begin{align}
&F\NE(\VNE{x}) + \Gamma_1^\top \mu + \Gamma_2^\top \lambda = 0, \label{eq:KKT_PEVs_stat} \\
&0 \le \mu \perp \gamma_1 - \Gamma_1 \VNE{x} \ge 0, \label{eq:KKT_PEVs_compl_individ} \\
&0 \le \lambda \perp \gamma_2 - \Gamma_2 \VNE{x} \ge 0 .\label{eq:KKT_PEVs_stat_compl_coupl}
\end{align}
\end{subequations}
Define $\tilde \mu$ and $\tilde \lambda$ as the dual variables corresponding to the active constraints (the other dual variables must be zero due to~\eqref{eq:KKT_PEVs_compl_individ} and~\eqref{eq:KKT_PEVs_stat_compl_coupl}).
The KKT system~\eqref{eq:KKT_PEVs}  in $\tilde \mu, \tilde \lambda$ only reads
\begin{equation}
\begin{aligned}
&\tilde \Gamma_1^\top \tilde \mu + \tilde \Gamma_2^\top \tilde \lambda = -F\NE(\VNE{x}) ,\\
&\tilde \mu, \tilde \lambda \ge 0 \,,
\end{aligned}
\label{eq:KKT_PEVs_reduced}
\end{equation}
where $\tilde \Gamma_1, \tilde \Gamma_2$ contain the subset of rows of $\Gamma_1, \Gamma_2$ corresponding to active constraints.
To conclude the proof we need to show that \eqref{eq:KKT_PEVs_reduced} has a unique solution $\tilde \lambda$.
To this end we apply the subsequent Lemma~\ref{lem:sufficient_for_uniqueness_subvector}.
To verify its assumption, we note that  its negation is equivalent, given the expressions of $\tilde \Gamma_1, \tilde \Gamma_2$ in~\eqref{eq:H_PEVs_partitioned}, to
the existence of $R' \subseteq R^\text{tight}$ such that for each vehicle $i$ it holds $\bar x_{\textup{N},t}^i  \in \{0,\tilde x^i_r\}$ for all $t \in R'$ or $\bar x_{\textup{N},t}^i  \in \{0,\tilde x^i_t\}$ for  $t\in\{1,\dots,n \}\setminus R'$ and  such $R'$ cannot exist by assumption. \hfill \smallblacksquare

\begin{lemma}
\label{lem:sufficient_for_uniqueness_subvector}
Consider $A_1 \in \R^{m \times n_1}$, $A_2 \in \R^{m \times n_2}$, $b \in \R^m$.
If the implication $A_1 x_1 + A_2 x_2 = 0 \; \Rightarrow \; x_1 = 0$ holds, then
the linear system of equations $A_1 x_1 + A_2 x_2 = b$
has at most one solution in $x_1$.
\end{lemma}
\indent \quad \textit{Proof:}
Assume $A \tilde x = b$ and $A \hat x = b$, then $A_1 \tilde{x}_1 + A_2 \tilde{x}_2 = b$ and $A_1 \hat{x}_1 + A_2 \hat{x}_2 = b$ imply 
$A_1(\hat{x}_1 - \tilde{x}_1) + A_2 (\hat{x}_2 - \tilde{x}_2) = 0$, which by assumption implies $\hat{x}_1 = \tilde{x}_1$. \hfill \smallblacksquare

\vspace*{-3mm}
\bibliographystyle{IEEEtran}
\bibliography{References}

\begin{IEEEbiography}[{\includegraphics[width=1.25in,height=1.25in,clip,keepaspectratio]{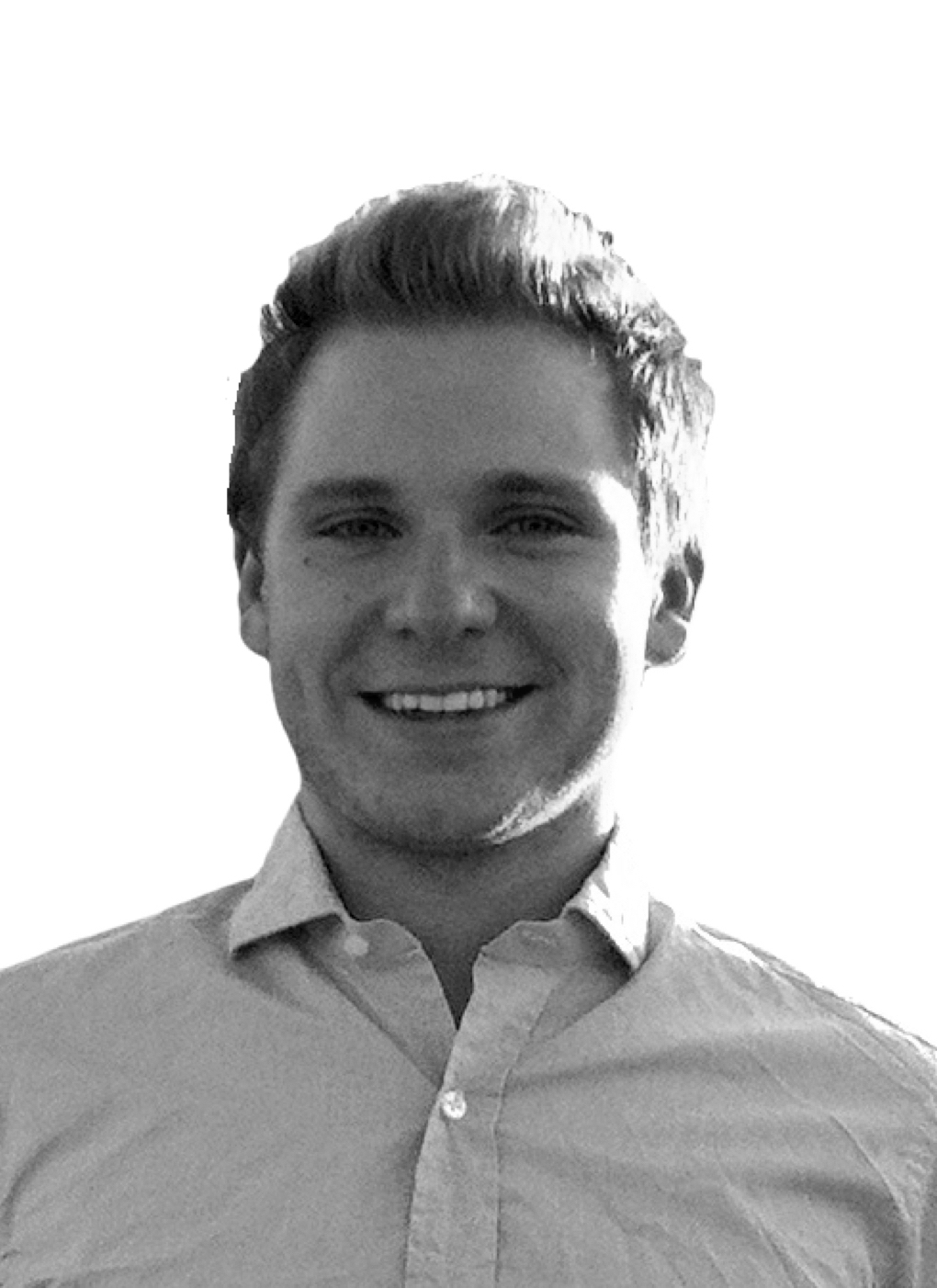}}]{Dario Paccagnan} is a doctoral student at the Automatic Control Laboratory, ETH Z\"{u}rich, Switzerland, since October 2014. 
He received his B.Sc. and M.Sc. in Aerospace Engineering from the University of Padova, Italy, in 2011 and 2014. In the same year he received the M.Sc. in Mathematical Modelling from the Technical University of Denmark, all with Honours. His Master's Thesis was prepared when visiting Imperial College of London, UK, in 2014. From March to August 2017 he has been a visiting scholar at the University of California, Santa Barbara. Dario's research interests are at the interface between distributed control and game theory. Applications include multiagent systems, smart cities and traffic networks.
\end{IEEEbiography}
\vspace*{-1cm}

\begin{IEEEbiography}[{\includegraphics[width=1.25in,height=1.25in,clip,keepaspectratio]{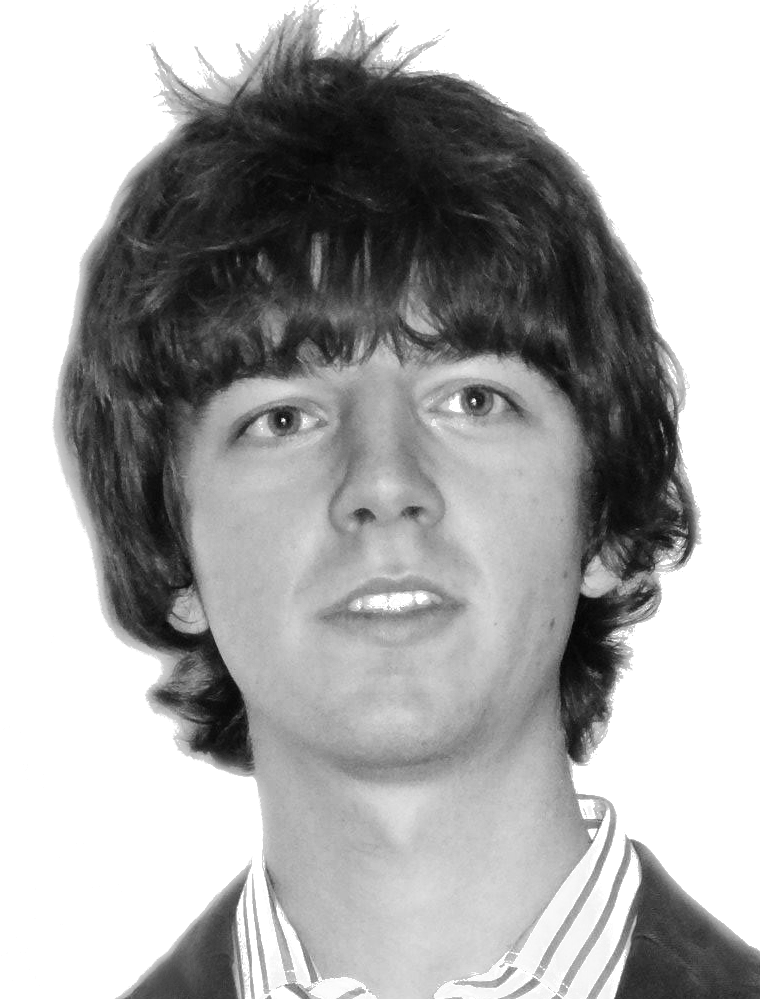}}]{Basilio Gentile} completed his PhD at the Automatic Control Laboratory at ETH Z\"{u}rich in 2018. He received his Bachelor's degree in Information Engineering and Master's degree in Automation Engineering from the University of Padova, as well as a Master's degree in Mathematical Modeling and Computation from the Technical University of Denmark. In 2013 he spent seven months in the Motion Lab at the University of California Santa Barbara to work at his Master's Thesis. His research focuses on aggregative games and network games with applications to traffic networks and to smart charging of electric vehicles.
\end{IEEEbiography}
\vspace*{-1cm}

\begin{IEEEbiography}[{\includegraphics[width=1.25in,height=1.25in,clip,keepaspectratio]{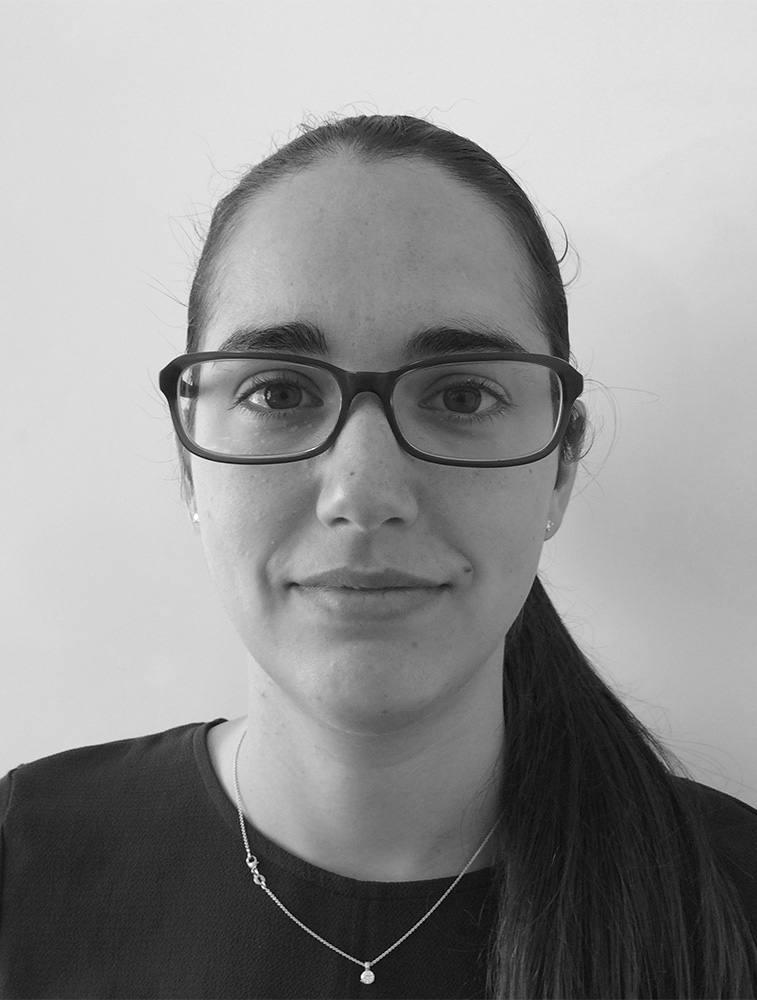}}]{Francesca Parise}
was born in Verona, Italy, in 1988. She received the B.Sc. and M.Sc. degrees (cum Laude) in Information and Automation Engineering from the University of Padova, Italy, in 2010 and 2012, respectively. She conducted her Master's thesis research at Imperial College London, UK, in 2012.  She graduated from the Galilean School of Excellence, University of Padova, Italy, in 2013. She defended her PhD at the Automatic Control Laboratory, ETH Z\"urich, Switzerland in 2016 and she is currently a Postdoctoral researcher at the Laboratory for Information and Decision Systems, M.I.T., USA.
Her research focuses on identification, analysis and control of complex systems, with application to distributed multi-agent networks and systems biology.
\end{IEEEbiography}
\vspace*{-1cm}

\begin{IEEEbiography}[{\includegraphics[width=1.25in,height=1.25in,clip,keepaspectratio]{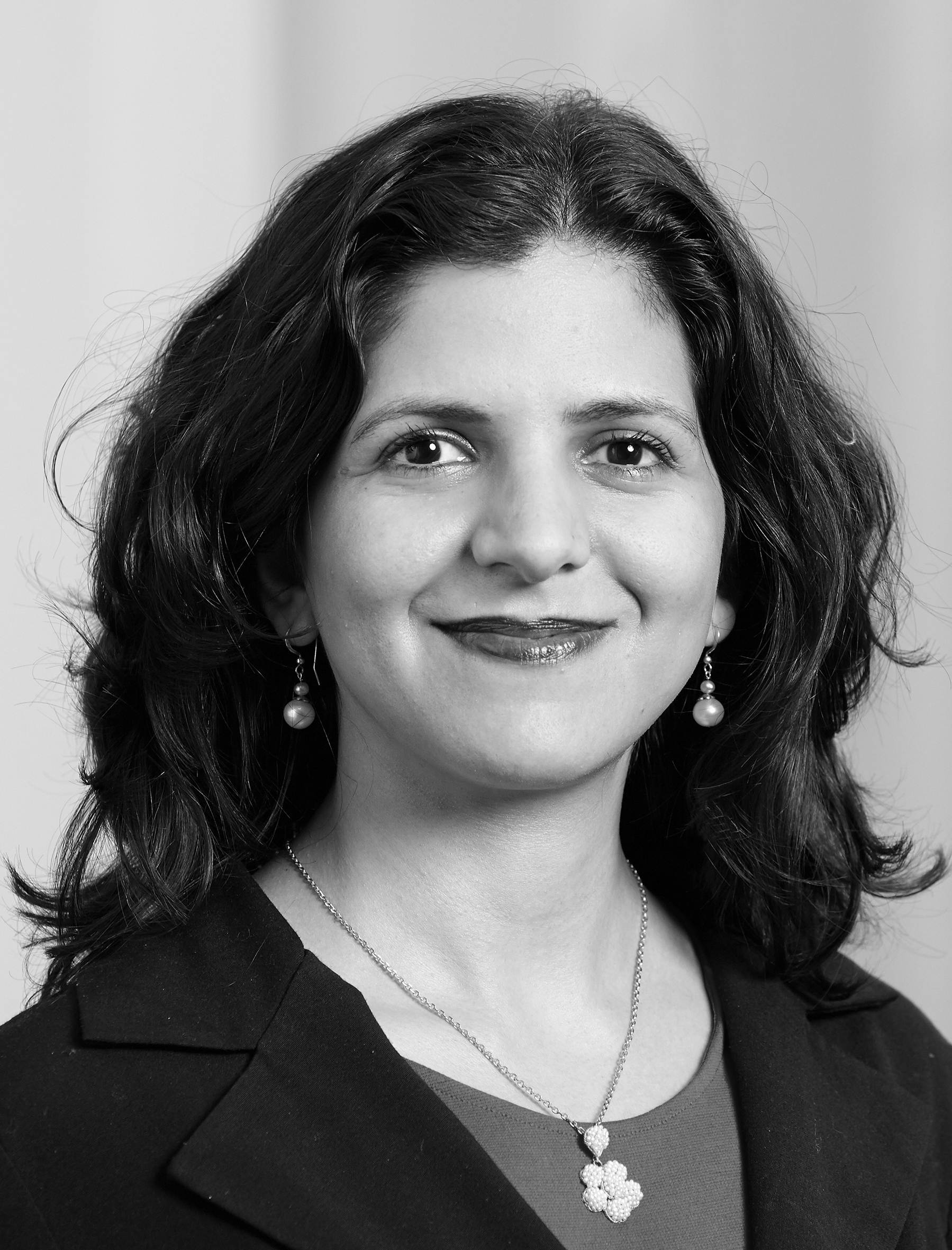}}]{Maryam Kamgarpour}
obtained her Master's and Ph.D. in Control Systems
at the University of California, Berkeley (2007, 2011) and her
Bachelor of Applied Sciences from University of Waterloo, Canada
(2005). Her research is on safety verification and optimal control of
large-scale uncertain dynamical systems with applications in air
traffic and power grid systems. She is the recipient of NASA High
Potential Individual Award, NASA Excellence in Publication Award
(2010) and the European Union (ERC) Starting Grant 2015.
\end{IEEEbiography}
\vspace*{-1cm}

\begin{IEEEbiography}[{\includegraphics[width=1.25in,height=1.25in,clip,keepaspectratio]{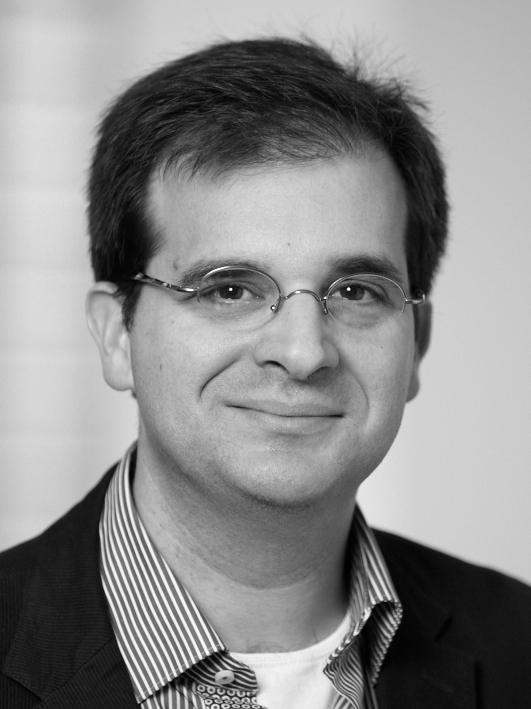}}]{John Lygeros}
completed a B.Eng. degree in electrical engineering in  1990 and an M.Sc. degree in Systems Control in 1991, both at Imperial College of Science Technology and Medicine, London, UK. In 1996 he obtained a Ph.D. degree from the Electrical Engineering and Computer Sciences Department, University of California, Berkeley. During the period 1996-2000 he held a series of research appointments. Between 2000 and 2003 he was a University Lecturer at the Department of Engineering, University of Cambridge, UK. Between 2003 and 2006 he was an Assistant Professor at the Department of Electrical and Computer Engineering, University of Patras, Greece. In July 2006 he joined the Automatic Control Laboratory at ETH Z\"urich, first as an Associate Professor, and since January 2010 as a Full Professor. Since 2009 he is serving as the Head of the Automatic Control Laboratory and since 2015 as the Head of the Department of Information Technology and Electrical Engineering. His research interests include modelling, analysis, and control of hierarchical, hybrid, and stochastic systems, with applications to biochemical networks, automated highway systems, air traffic management, power grids and camera networks. John Lygeros is a Fellow of the IEEE, and a member of the IET and the Technical Chamber of Greece.
\end{IEEEbiography}

\end{document}